\documentclass[aps,prb,10pt,twocolumn,superscriptaddress]{revtex4-2}

\usepackage{graphicx}
\usepackage[linesnumbered,ruled,vlined]{algorithm2e}
\SetKwInput{KwRequire}{Require}
\SetKwInput{KwReturn}{Return}
\usepackage{xcolor}
\usepackage{physics}
\usepackage{tikz}
\usetikzlibrary{shapes.geometric}
\usepackage{dsfont}
\usepackage{booktabs}
\usepackage{amsfonts}
\usepackage{amsmath}
\usepackage{amssymb}
\usepackage{amsthm}
\usepackage{quantikz}
\usepackage{tcolorbox}
\usepackage{enumerate}
\usepackage{orcidlink}
\usepackage{bm}
\usepackage{mathtools}
\usepackage{braket}
\usepackage[right]{lineno}
\usepackage{listings}
\usepackage{subcaption}
\usepackage{caption}
\usepackage{enumitem}
\usepackage{bbm}
\usepackage{array}

\usepackage{hyperref}
\usepackage{todonotes}

\usepackage{hyperref}
\hypersetup{hidelinks}
\usepackage[capitalize,nameinlink]{cleveref}

\newtheorem{theorem}{Theorem}

\theoremstyle{definition}

\crefname{equation}{Eq.}{Eqs.}
\Crefname{equation}{Eq.}{Eqs.}
\crefname{section}{Sec.}{Secs.}
\Crefname{section}{Sec.}{Secs.}
\crefname{figure}{Fig.}{Figs.}
\Crefname{figure}{Fig.}{Figs.}
\crefname{appendix}{Appendix}{Appendices}
\Crefname{appendix}{Appendix}{Appendices}
\crefname{table}{Table}{Tables}
\Crefname{table}{Table}{Tables}
\crefname{theorem}{Theorem}{Theorems}
\Crefname{theorem}{Theorem}{Theorems}
\crefname{lemma}{Lemma}{Lemmas}
\Crefname{lemma}{Lemma}{Lemmas}
\crefname{corollary}{Corollary}{Corollaries}
\Crefname{corollary}{Corollary}{Corollaries}
\crefname{definition}{Definition}{Definitions}
\Crefname{definition}{Definition}{Definitions}
\crefname{algocf}{Algorithm}{Algorithms}
\Crefname{algocf}{Algorithm}{Algorithms}

\newcommand{\bs}{\boldsymbol}

\newcommand{\id}{\mathds{1}}                 
\newcommand{\idtwo}{\mathds{1}_{2}}

\newcommand{\idrest}{\mathds{1}_{\backslash \ell}}

\begin{document}

\title{Noisy quantum circuit simulation with the tensor jump method}

\author{Maximilian Fröhlich\,\orcidlink{0009-0007-5276-2858}}
\email{froehlich@wias-berlin.de}
\affiliation{Weierstrass Institute for Applied Analysis and Stochastics, Berlin, Germany}

\author{Aaron Sander\,\orcidlink{0009-0007-9166-6113}}
\email{aaron.sander@tum.de}
\affiliation{Technical University of Munich, Munich, Germany}

\author{Martin Eigel\,\orcidlink{0000-0003-2687-4497}}
\affiliation{Weierstrass Institute for Applied Analysis and Stochastics, Berlin, Germany}

\author{Robert Wille\,\orcidlink{0000-0002-4993-7860}}
\affiliation{Technical University of Munich, Munich, Germany}
\affiliation{Munich Quantum Software Company GmbH, Munich, Germany}
\affiliation{Software Competence Center Hagenberg GmbH (SCCH), Hagenberg, Austria}

\author{Michael Hintermüller\,\orcidlink{0000-0001-9471-2479}}
\affiliation{Weierstrass Institute for Applied Analysis and Stochastics, Berlin, Germany}

\date{June 29, 2026}

\begin{abstract}
Classical simulation of noisy quantum circuits is essential for validating algorithms, benchmarking hardware, and guiding error-mitigation strategies. Yet it remains challenging due to either the exponential memory footprint of density matrix methods or the high variance in trajectories of standard Kraus-insertion approaches. We introduce a variance aware tensor network framework that unifies the tensor jump method (TJM) with a local time-dependent variational principle (TDVP) and a sparse Pauli-Lindblad model (SPLM) of hardware noise. Gates are realized as short variational evolutions on the matrix product state (MPS) manifold, while noise is applied stochastically per circuit window using Pauli-Lindblad jump sets that render the sampling hazard state-independent and the dissipative contraction trivial after renormalization. Strikingly, our method supports arbitrary Lindblad correlated multi-qubit noise channels consistent with hardware connectivity, including long-range noise operators on non-adjacent subsets, enabling direct simulation of crosstalk and other connectivity-induced errors beyond local noise models. Building on prior work that introduced analog unitary unravelings, we adapt this scheme to a framework based on the Lindblad master equation with exact generator matching under discretized angle laws and a parameter choice. This exactly matches the Lindblad generator, while selecting Gaussian or two-point angle laws. In parallel, we introduce an MPS-projector-jump unraveling that yields state-independent hazards and closed-form variance laws. Both schemes are unbiased and inherit the standard $1/\sqrt{N}$ Monte Carlo convergence, but with substantially smaller constants due to variance reduction. Empirically, projector sampling markedly reduces the required bond dimension per trajectory across many circuit architectures, while analog sampling excels at low noise. We demonstrate accurate, scalable noisy-circuit simulation on a 25-qubit noisy XY quench and IBM's 127-qubit kicked Ising benchmark with long-range depolarizing noise, achieving substantially reduced Monte Carlo variance and favorable bond dimension growth compared to standard Kraus insertion baselines.
\end{abstract}

\maketitle

\medskip

\section{Introduction}
\label{sec:intro}

Classical simulation of quantum circuits remains a central ingredient in quantum computing: it supports algorithm verification, compiler validation, device benchmarking, and the quantitative assessment of noise and error-mitigation strategies. While noiseless circuit simulation has seen dramatic progress, from stabilizer-based techniques for structured circuits to tensor network and Pauli-propagation approaches for generic architectures \cite{angrisani2025simulatingquantumcircuitsarbitrary,Aaronson_2004,Vidal_2003,rudolph_pauli_propagation_2024,sander2025quantumcircuitsimulationlocal}, the noisy setting is substantially more challenging because noise turns pure state evolution into an effectively mixed state process.

A standard route to exact noisy simulation is density-matrix propagation, applying gates as $\rho \mapsto U\rho U^\dagger$ and noise as a CPTP map in Kraus form $\rho \mapsto \sum_\alpha K_\alpha \rho K_\alpha^\dagger$ \cite{qiskit2024}. This is fully general but requires storing a $2^n\times 2^n$ operator and therefore scales as $\mathcal{O}(4^n)$ in memory (and analogously, time). Operator compression via matrix product operators/density operators (MPO/MPDO) can soften this barrier, but for generic circuits and correlated noise the required operator bond dimension can still grow rapidly, so practical accuracy becomes a cost-truncation tradeoff \cite{zwolak_vidal_2004,verstraete_mpdo_2004}.

Important large-scale specializations exist when circuits or noise are structured. For Clifford(-dominated) circuits with Pauli-stochastic noise, stabilizer and tableau simulators can scale to very large system sizes, but do not natively cover generic non-Clifford resources or coherent and non-Pauli noise without additional overhead or approximation \cite{Aaronson_2004,stim_2021}. In the same spirit, Pauli-propagation methods evolve observables backward through gates and adjoint noise channels; they can be extremely efficient when the Pauli expansion remains sparse, but can suffer rapid term growth with depth, non-Clifford mixing, or strongly correlating multi-qubit noise \cite{rudolph_pauli_propagation_2024,rudolph_pauli_propagation_noise_2025}.

A complementary strategy is trajectory sampling via the Monte Carlo wave function method (MCWF), where the mixed state evolution is recovered from an ensemble of pure state trajectories generated by non-Hermitian drift interspersed with stochastic jumps \cite{PlenioKnight,dalibard_castin_molmer_1992}. Each trajectory inherits the favorable $\mathcal{O}(2^n)$ state-vector scaling and can be combined with MPS compression, but efficiency is often limited by trajectory entanglement growth and by Monte Carlo variance, which can force a large number of samples even when each trajectory is individually cheap.

In this work we develop a trajectory-based tensor network framework for noisy quantum circuit simulation that unifies (i) local time-dependent variational principle (local TDVP) gate application on the MPS manifold \cite{sander2025quantumcircuitsimulationlocal} with (ii) tensor network MCWF sampling in the spirit of the tensor jump method (TJM) \cite{sander2025largescalestochasticsimulationopen}. Relative to Pauli propagation, the method evolves full (pure) trajectory states so many observables can be extracted from the same trajectories without observable-specific branching overhead \cite{rudolph_pauli_propagation_2024,rudolph_pauli_propagation_noise_2025}. Relative to density-matrix MPO/MPDO approaches, it trades determinism for memory efficiency by representing the mixed state as an ensemble of cheaper MPS trajectories \cite{zwolak_vidal_2004,verstraete_mpdo_2004}.

To obtain a compact, device-calibrated description of realistic noise, we focus on the sparse Pauli-Lindblad model (SPLM), which represents each layer's noise as the exponential of a Lindblad generator supported on a sparse set of low-weight Pauli strings \cite{vandenBerg2024techniqueslearning,van_den_Berg_2023}. This choice is algorithmically decisive: for Pauli jumps $L_m=P_m$ one has $P_m^\dagger P_m=\id$, so jump hazards become state-independent and the dissipative contraction collapses to a global scalar that cancels upon renormalization; jump probabilities can therefore be precomputed per layer. Moreover, coherent hardware errors can be Pauli-tailored into stochastic Pauli noise via randomized compiling/Pauli twirling, so focusing on SPLM is often a pragmatic modeling choice rather than a fundamental restriction \cite{WallmanEmerson2016randomizedcompiling}.

Beyond local single- and nearest-neighbor noise, calibrated hardware models can include correlated errors consistent with device connectivity, including effective two-qubit channels on non-adjacent subsets. Such long-range correlated noise can be cumbersome for many simulators because it requires SWAP-based locality reductions or induces rapid operator growth, whereas in our setting the relevant collapse operators reduce (up to a scalar) to $a\,\id+b\,P$, admitting an exact bond-$2$ MPO independent of separation (Sec.~\ref{sec:longrange-mpo}). Finally, we exploit unraveling freedom in Lindblad dynamics: recent work uses state-dependent mixing to minimize trajectory entanglement \cite{VovkPichler_PRL_2022,VovkPichler_PRA_2024}, while here we develop two complementary variance-aware unravelings tailored to Pauli-Lindblad terms, an analog unitary-mixture unraveling with exact generator matching \cite{AnalogSimQuantinuum,Chen_2024} and a projector-jump unraveling with closed-form variance laws on absorbing circuit windows (Sec.~\ref{sec:proj_unraveling}).

We demonstrate the resulting circuit Tensor Jump Method (cTJM) on many-body circuit benchmarks, empirically identifying a practical regime separation: analog sampling is most effective at weak noise, whereas projector sampling is advantageous at moderate-to-strong noise by reducing both estimator variance and bond dimension growth.

\subsection{Organization}
In \cref{sec:localTDVP,sec:tjm,sec:splm} we review local TDVP circuit evolution, the TJM for open systems, and the SPLM noise model.
\Cref{sec:cTJM} introduces cTJM.
\Cref{sec:analog-unraveling} and \cref{theorem:projector-variance} develop the analog and projector unravelings and their variance properties, including the long-range MPO construction in \cref{sec:longrange-mpo}.
\Cref{sec:variance-check} validates the variance predictions in a controlled two-qubit setting, and \cref{sec:xy-quench-bitflip,sec:127-qubits} present many-body circuit benchmarks of 25 and 127 qubits with various noise models.
We conclude with a discussion of hybrid unraveling strategies and variance-driven adaptive choices.

\section{Background}

In this section we summarize the main ideas of the quantum circuit simulation with the local TDVP and the TJM for open quantum systems \cite{sander2025largescalestochasticsimulationopen, sander2025quantumcircuitsimulationlocal}, which are then combined and modified to the cTJM in \cref{sec:cTJM}.

\subsection{Local Time-Dependent Variational Principle for Quantum Circuits}
\label{sec:localTDVP}

The \emph{local Time-Dependent Variational Principle} (local TDVP) is a variational method for simulating quantum circuit evolution within the MPS manifold. Unlike TEBD, which applies discrete gate unitaries with potential overhead from SWAP networks, circuit TDVP interprets each multi-qubit gate $U = e^{-iH}$ as the exponential of a local Hamiltonian $H$ and simulates its action by integrating the corresponding Schrödinger equation, projected onto the MPS manifold.

At each gate application, the time derivative is projected onto the local tangent space of the current MPS, ensuring that the evolution remains variationally optimal for the chosen bond dimension. This approach naturally accommodates gates between non-neighboring qubits, thereby eliminating the need for explicit SWAP operations and enabling efficient simulation of circuits with arbitrary connectivity.

Throughout the evolution, the MPS bond dimension can be dynamically adjusted, providing a balance between computational cost and simulation accuracy, and preventing the uncontrolled growth in entanglement characteristic of TEBD-based circuit simulations. In \cite{sander2025quantumcircuitsimulationlocal} we have shown that there are cases in which local TDVP uses the MPS parameters more efficiently than TEBD, such that the bond dimension growth is much slower while having the same complexity and accuracy.

\subsection{Tensor Jump Method}
\label{sec:tjm}
In this section, the main idea of the TJM is recalled, which was recently proposed in \cite{sander2025largescalestochasticsimulationopen}.
We first recall the high-level idea, which is basically the MCWF method in tensor format.
The goal is to simulate the Lindblad master equation (LME)
\begin{equation}
  \label{eq:lme}
  \frac{d\rho}{dt} \;=\; -i[H,\rho(t)] + \sum_{m=1}^k \gamma_m \!\left(L_m \rho L_m^\dagger - \tfrac12\{L_m^\dagger L_m,\rho(t)\}\right).
\end{equation}
For any observable $O$, the trajectory estimator $X_t=\langle\psi_t|O|\psi_t\rangle$ satisfies \\
$\frac{1}{N}\sum_{i=1}^N X_t^{(i)} \to \Tr(O\rho(t))$ with Monte Carlo error $\sim \sqrt{\text{Var}(X_t)/N}$.

\subsubsection{One time step of TJM}
\label{sec:tjm-step}
Fix a step size $\delta t$ and let $|\psi\rangle$ be the (normalized) input MPS for this step. Define the non-Hermitian effective Hamiltonian
\begin{equation}
  \label{eq:Heff}
  H_{\mathrm{eff}} \;=\; H - \tfrac{i}{2}\sum_{m=1}^k \gamma_m L_m^\dagger L_m .
\end{equation}
The TJM realizes the drift $e^{-iH_{\mathrm{eff}}\delta t}$ via Strang splitting into a Hermitian unitary step and a dissipative contraction,
\begin{equation}
\begin{aligned}
  \label{eq:strang-splitting}
  &e^{-i H_{\mathrm{eff}}\delta t} \\
  \;&=\;
  \underbrace{e^{-\tfrac{1}{2}\delta t \sum_m \gamma_m L_m^\dagger L_m}}_{:=\,\mathcal{D}[\delta t/2]}
  \;\underbrace{e^{-i H\,\delta t}}_{:=\,\mathcal{U}[\delta t]}
  \;\underbrace{e^{-\tfrac{1}{2}\delta t \sum_m \gamma_m L_m^\dagger L_m}}_{:=\,\mathcal{D}[\delta t/2]}
  \;+\; \mathcal{O}(\delta t^3).
  \end{aligned}
\end{equation}
Accordingly, the unnormalized no-jump update is
\begin{equation}
  \label{eq:drift}
  |\tilde\psi\rangle
  \;=\;
  \mathcal{D}[\delta t/2]\;\mathcal{U}[\delta t]\;\mathcal{D}[\delta t/2]\;|\psi\rangle .
\end{equation}
In MPS form, $\mathcal{U}[\delta t]$ is applied variationally (e.g., via 1-TDVP/2-TDVP) on the MPS manifold, while $\mathcal{D}[\delta t/2]$ is a non-unitary contraction \cite{Haegeman_2013,Haegeman_2016,Haegeman2011,sander2025largescalestochasticsimulationopen}.

The total jump probability is determined by the norm loss,
\begin{equation}
  \label{eq:delta-p}
  \delta p \;=\; 1-\|\tilde\psi\|^2
  \;\approx\;
  \delta t \sum_{m=1}^k \gamma_m \langle\psi|L_m^\dagger L_m|\psi\rangle ,
\end{equation}
where the approximation is the standard first-order MCWF hazard for small $\delta t$.
With probability $1-\delta p$ set $|\psi'\rangle = |\tilde\psi\rangle/\|\tilde\psi\|$ (no jump). Otherwise, sample a channel index $m$ with
\begin{equation}
  \label{eq:jump-prob}
  \mathbb{P}[m\mid\text{jump}]
  \;=\;
  \frac{\gamma_m \langle\psi| L_m^\dagger L_m |\psi\rangle}{\sum_{j=1}^k \gamma_j \langle\psi| L_j^\dagger L_j |\psi\rangle}
\end{equation}
and apply the normalized jump
\begin{equation}
  \label{eq:jump-sampling}
  |\psi'\rangle \;=\; \frac{L_m|\psi\rangle}{\|L_m|\psi\rangle\|}.
\end{equation}
In tensor form, $L_m$ is applied by contracting the corresponding local operator into the affected site(s), followed by re-canonicalization.

\subsection{Sparse Pauli-Lindblad Model}
\label{sec:splm}

A key ingredient for scalable simulation of noisy quantum circuits is a compact yet accurate representation of the noise. The \emph{Sparse Pauli-Lindblad Model} (SPLM)~\cite{van_den_Berg_2023} provides such a representation by describing the noise channel in each circuit layer as the exponential of a Lindblad generator with only a sparse set of local Pauli terms.

\subsubsection{Definition}

Consider an $n$-qubit noise channel $\Lambda$ arising from a sparse set of local noise interactions. The SPLM assumes that the channel is Markovian and can be written in Lindblad form with Pauli jump operators
\begin{equation}
  \mathcal{L}(\rho) = \sum_{m=1}^k \gamma_m \left( P_m \rho P_m^\dagger - \rho \right),
  \label{eq:splm_generator}
\end{equation}
with $P_m \in \{\idtwo,X,Y,Z\}^{\otimes n}$, $\gamma_m \ge 0$ are non-negative rates, and the support of each $P_m$ is restricted to match the device connectivity (typically weight-one and weight-two Paulis).

The corresponding noise channel is
\begin{equation}
  \Lambda(\rho) = e^{\mathcal{L}}(\rho)
  = \bigcirc_{m=1}^k \left( w_m \, \rho + (1-w_m) \, P_m \rho P_m^\dagger \right),
  \label{eq:splm_channel}
\end{equation}
with
\begin{equation}
  w_m = \frac{1 + e^{-2\gamma_m}}{2},
\end{equation}
where $\bigcirc_{m=1}^k$ denotes the composition of $k$ Kraus channels.
Since all $P_m$ commute up to a phase, the factors in~\eqref{eq:splm_channel} commute and can be applied in any order.

\subsubsection{Motivation for combination with TJM}
\label{sec:motivation}
In the TJM formulation of \cref{sec:tjm}, the Lindblad part is written in terms of
generic jump operators $L_m$ and rates $\gamma_m$.
For an SPLM channel, we simply take
\[
  L_m := P_m
\]
so that \eqref{eq:lme} reduces to the Pauli-Lindblad form \eqref{eq:splm_generator}.
In particular,
\begin{equation}
  P_m P_m^\dagger = \id \quad \forall\,m,
\end{equation}
which is exploited to simplify both the jump sampling and the dissipative contraction.

Consequently, we get that
\begin{equation}
  \delta p_m = \delta t \gamma_m \langle \psi(t) | P_m^\dagger P_m | \psi(t) \rangle = \delta t \gamma_m \langle \psi(t)| \psi(t) \rangle \quad \forall m.
\end{equation}
Thus the probability distribution stays the same for all $n$ timesteps.
Instead of calculating it again, it simply has to be calculated once before the first timestep.
With this we also do not need the stochastic MPS $\Phi$ anymore, which is an auxiliary MPS used in \cite{sander2025largescalestochasticsimulationopen} to sample the jump operators in each timestep.
Since the sampling is no longer time dependent, we can precompute the fixed jump probabilities
\[
  p_m \;=\; \frac{\gamma_m}{\sum_{j=1}^k \gamma_j}, \quad m=1,\dots,k,
\]
store the probability vector $\bs p := (p_1,\dots,p_k)$, and sample an index $m$ from $\bs p$ in every timestep where $\epsilon < \delta p$.
A third simplification provided by the SPLM concerns the dissipative contraction:
\begin{equation}
  \label{eq:DissipativeMPO}
  \begin{aligned}
    \mathcal{D}[\delta t]
    &= \exp\!\left[-i\left(-\frac{i}{2}\sum_{\ell=1}^{L}\sum_{j\in S(\ell)}\gamma_j\,(\idrest \otimes L_j^{[\ell]\dagger}L_j^{[\ell]})\right)\delta t\right] \\
    &= \prod_{\ell=1}^{L}\exp\!\left(-\frac{1}{2}\delta t\sum_{j\in S(\ell)}\gamma_j\,(\idrest \otimes L_j^{[\ell]\dagger}L_j^{[\ell]})\right)\\
    &= \prod_{\ell=1}^{L}\exp\!\left(\idrest \otimes\left[-\frac{1}{2}\delta t\sum_{j\in S(\ell)}\gamma_j\,L_j^{[\ell]\dagger}L_j^{[\ell]}\right]\right) \\
    &= \prod_{\ell=1}^{L}\exp\!\left( -\tfrac{1}{2}\Gamma_\ell \delta t \, (\idrest \otimes \idtwo) \right)
    = e^{-\frac{1}{2}\delta t \sum_{\ell=1}^{L} \Gamma_\ell}\, \id \\
    &= e^{-\frac{1}{2}\delta t \Gamma_{\mathrm{tot}}}\, \id,
  \end{aligned}
\end{equation}
where we have set $\Gamma_\ell:=\sum_{j\in S(\ell)}\gamma_j,\; L_j^{[\ell]\dagger}L_j^{[\ell]}=\idtwo$ and $S(\ell) \subset \{1,\ldots,k\}$ is the set of jump operators that affect site $\ell$. Thus the dissipative contraction reduces to a simple scaling of the site tensors, which is physically irrelevant after renormalization of the state. It only serves the norm reduction to sample the jump operators.

\section{Circuit Tensor Jump Method}
\label{sec:cTJM}
We adapt the TJM to noisy quantum circuits by interleaving its trajectory sampling using the dissipative contraction and the stochastic jump rule with the local TDVP update for quantum circuits described in \cref{sec:localTDVP}.

Let $\mathcal{U}=(U_1,\dots,U_M)$ be a circuit on $n$ qubits. $U_g$ acts on the qubit set $Q_g$, where we restrict ourselves to $|Q_g| \in \{1,2\}$.
Noise is given by an arbitrary Markovian Lindblad model with jump set $\{(\gamma_m,L_m)\}_{m=1}^k$ where each $L_m$ is an arbitrary Lindblad operator that has finite support in the qubit line.
For the gate $U_g$ we restrict to the local generator that only acts on the qubits $Q_g$ and contains only the local jump set which we define as $S_g := \{\, m \in \{1,\dots,k\} \;\vert\; \mathrm{supp}(L_m)\cap Q_g\neq\emptyset \,\}$.
Thus we can define the generator of the local noise as
\begin{equation}
  \label{eq:cTJM-local-L}
  \mathcal{L}_g(\rho)
  =\sum_{m \in S_g}\!
  \gamma_m\!\left(L_m\rho L_m^\dagger-\tfrac12\{L_m^\dagger L_m,\rho\}\right),
\end{equation}
so that only channels that touch $Q_g$ are considered in the window. To avoid confusion, we explicitly denote by $S(\ell)$ the set of all jump operators acting on site $\ell$, and by $S_g$ those whose support intersects the gate support $Q_g$.
In particular, for a single-qubit gate $U_g$ on site $\ell$ (i.e.\ $Q_g=\{\ell\}$) we have $S_g = S(\ell)$ and for multi-qubit gates we have $S_g = \cup_{\ell \in Q_g}S(\ell)$.

The deterministic part in the cTJM step then works as follows:
We write the trajectory state $|\psi\rangle$ as an MPS, apply $U_g=\exp(-i H_g)$ variationally on the local tangent space of the tensors in $Q_g$ (1-TDVP or 2-TDVP), re-canonicalize, and truncate if needed. TDVP works as an orthogonal projector on the manifold of a given bond dimension and makes it possible to apply long-range gates without SWAP-gate insertion \cite{sander2025quantumcircuitsimulationlocal}.

After the TDVP gate update, we apply the dissipative contraction
$\mathcal{D}[\delta t]$ followed by jump sampling
$\mathcal{J}_\epsilon[\delta t]$. The maps $\mathcal{D}$ and
$\mathcal{J}_\epsilon$ are exactly the same as in the original TJM scheme
described in \cref{sec:tjm}; the only difference is that in cTJM we apply
them once per gate with a full time step $\delta t$ and omit the half-step
contractions.
We allow for at most one jump per 2-qubit gate per trajectory, which could be generalized to multi-jump events. Nevertheless this is out of the scope of this work.
Finally, we recover the circuit-level Lindblad expectations by the same trajectory averaging as in \cite{sander2025largescalestochasticsimulationopen}
\[
  \widehat{\langle O\rangle}(g)\;=\;\frac{1}{N}\sum_{i=1}^N \bra{\psi^{(i)}_g} O \ket{\psi^{(i)}_g}\,,
\]
evaluated at the desired gate index $g$ (or measurement positions).

\subsection{Algorithm sketch (one trajectory)}
\label{sec:ctjm-algorithm}
We mirror the TJM step gate-by-gate: after applying each circuit gate via local TDVP on its support $Q_g$, we apply the same dissipative contraction and jump sampling as in \cite{sander2025largescalestochasticsimulationopen}, but restricted to the local jump set $S_g$. Illustrated in \cref{alg:cTJM}.

\begin{algorithm}[t]
\caption{cTJM (one trajectory)}
\label{alg:cTJM}

\KwRequire{Circuit $\mathcal{U}=(U_1,\ldots,U_M)$; Lindblad jumps $\{(\gamma_m,L_m)\}_{m=1}^k$}

Initialize MPS $\ket{\psi^{(i)}_0}\leftarrow \ket{\psi_0}$\;

\For{$g=1$ \KwTo $M$}{
  $Q_g \leftarrow \mathrm{supp}(U_g)$;\quad
  $S_g \leftarrow \{\, m : \mathrm{supp}(L_m)\cap Q_g \neq \emptyset \,\}$\;

  \textbf{local TDVP gate:}
  $\ket{\phi}\leftarrow
  \mathrm{TDVP}\bigl(\ket{\psi^{(i)}_{g-1}},\, U_g \text{ on } Q_g\bigr)$;
  re-canonicalize / truncate\;

  \textbf{Local dissipative contraction:}
  $\ket{\tilde\psi}\leftarrow \mathcal D_g\,\ket{\phi}$
  \textit{(same $\mathcal D[\cdot]$ as in \cite{sander2025largescalestochasticsimulationopen}, but with $m\in S_g$ only)}\;

  \textbf{Local jump:}
  $\ket{\psi^{(i)}_g}\leftarrow \mathrm{JumpSampling}(\ket{\tilde\psi})$
  \textit{(same as \cite{sander2025largescalestochasticsimulationopen}, but with $\mathcal{L}_g$ only)}\;

  Optionally record
  $X^{(i)}_g=\bra{\psi^{(i)}_g} O \ket{\psi^{(i)}_g}$\;
}

\KwReturn{trajectory $\{\ket{\psi^{(i)}_g}\}_{g=0}^M$; ensemble means via averaging}

\end{algorithm}

\subsection{Error analysis of cTJM}
\label{sec:ctjm-error}
The total error of cTJM splits into a stochastic (Monte Carlo) component and deterministic approximation biases.
The sampling uncertainty is the standard Monte Carlo error
\begin{equation}
\mathrm{SE}\!\left[\widehat{\langle O\rangle}_g\right]
=\sqrt{\frac{\operatorname{Var}(X_g)}{N_{\mathrm{traj}}}},
\end{equation}
so variance reduction directly decreases the prefactor of the $1/\sqrt{N_{\mathrm{traj}}}$ convergence
\cite{sander2025largescalestochasticsimulationopen}.
Deterministic errors arise from the local TDVP realization of gates on the MPS manifold.
Concretely, besides the projection error incurred by restricting the evolution to a finite bond dimension manifold,
there is a TDVP time-integration error of order $O(\delta t^3)$ per (sub)step and $O(\delta t^2)$ over a fixed
evolution interval \cite{sander2025largescalestochasticsimulationopen}.
In addition, interleaving the coherent gate update with the dissipative/noise step constitutes an operator splitting,
whose (Strang-type) time-step error scales as $O(\delta t^3)$ \cite{Strang1968,sander2025largescalestochasticsimulationopen}.
Both the TDVP integration and splitting errors can be systematically reduced by subdividing a gate/window into smaller
effective steps, i.e., by performing multiple TDVP sweeps (substeps) per gate.
Finally, MPS SVD compression introduces a truncation error controlled by the threshold and
the cap $\chi_{\max}$, which can be monitored via discarded weight.
The only error that cannot be estimated exactly in general is the projection error incurred by restricting the dynamics to a lower-dimensional manifold. For TDVP it is the residual
\begin{equation}
\epsilon(\chi) = \left\|\bigl(\id-P_{\mathcal M_\chi,\ket{\Phi}}\bigr)\,H_0\ket{\Phi}\right\|_2^2,
\end{equation}
where $H_0$ is a Hamiltonian, $P_{\mathcal M_\chi,\ket{\Phi}}$ the TDVP projector onto the tangent space of manifold $\mathcal{M}_{\chi}$ at the quantum state $\ket{\Phi}$.
Note that it vanishes for strictly nearest-neighbor Hamiltonians for $2$TDVP, but beyond such special cases there is no general-purpose estimator for its accumulated impact on observables \cite{sander2025largescalestochasticsimulationopen, Haegeman_2014}.

\section{Variance-aware unravelings}

In the sparse Pauli-Lindblad model of \cref{sec:splm}, the Lindblad generator
\begin{equation}
\mathcal{L}(\rho) = \sum_{m=1}^k \gamma_m \big( P_m \rho P_m^\dagger - \rho \big)
\label{eq:splm_generator_recall}
\end{equation}
is a sum of single Pauli-string channels, each of the form $\dot\rho = \gamma_m (P_m \rho P_m - \rho)$.
Throughout this section we therefore fix one such term $(\gamma,P)$ and
construct different unravelings (choices of collapse operators)
that all reproduce the same Pauli-Lindblad contribution
\begin{equation*}
\dot\rho = \gamma (P\rho P - \rho),
\end{equation*}
but have different variances and average bond dimensions in the MPS.
In the full SPLM, these constructions are applied channel-wise with
$(\gamma,P)\leftarrow(\gamma_m,P_m)$ (or to groups of channels) so that the
total generator remains \eqref{eq:splm_generator}.

For any set of active noise channels, we write
\[
\Gamma_{\mathrm{tot}} := \sum_{k\in\mathcal K}\gamma_k,
\]
where $\mathcal{K}$ denotes the general set of jump operators which is dependent on the unraveling method and $|\mathcal{K}|=m$ may not be true. Given a readout observable $O$, we further split
\begin{equation*}
\begin{aligned}
&\Gamma_{\mathrm{anti}} := \sum_{\substack{k\in\mathcal K\\ \{O,P_k\}=0}}\gamma_k,\\
&\Gamma_{\mathrm{comm}} := \sum_{\substack{k\in\mathcal K\\ [O,P_k]=0}}\gamma_k, \\
&\Gamma_{\mathrm{tot}}=\Gamma_{\mathrm{anti}}+\Gamma_{\mathrm{comm}}.
\end{aligned}
\end{equation*}

The Pauli-Lindblad term $\dot\rho=\gamma(P\rho P-\rho)$ admits many equivalent unravelings.
Different choices of collapse operators yield the same Lindblad generator but different trajectory statistics.
Standard sampling realizes this channel via random Pauli flips $P$.
This is unbiased but can lead to large trajectory-to-trajectory fluctuations.

In the following, we introduce two alternative unravelings: (i) the analog unraveling, previously used in the Kraus-operator sampling framework of Ref. \cite{AnalogSimQuantinuum}, and (ii) the projector unraveling, standard in the quantum-trajectory literature but, to our knowledge, not previously explored as a variance- and bond dimension–reduction strategy for MPS-based trajectory simulations.

\subsection{Analog sampling}
\label{sec:analog-unraveling}

For a Pauli-Lindblad term on a (possibly multi-qubit) Pauli string $P$ given by
\begin{equation}
\label{eq:analog-lindblad}
\frac{d\rho}{dt}  \;=\; \gamma\,(P\rho P - \rho),\quad P^\dagger=P,\; P^2=\id,
\end{equation}
we use a unitary analog mixture with rate-free collapse operators
\[
L_\theta := e^{i\theta P},\quad \theta\in\mathbb{T}_\pi:=\mathbb{R}/(\pi\mathbb{Z})\,,
\]
where the $\pi$-periodicity follows from $P^2=\id\Rightarrow e^{i(\theta+\pi)P}=-e^{i\theta P}$.
The most common choice for $\mathbb{T}_\pi$ is $\big(-\frac{\pi}{2},\tfrac{\pi}{2}]$.
Furthermore we define a probability law
\[
w:\mathbb{T}_\pi\to[0,\infty),\quad w(\theta)=w(-\theta),\quad \int_{\mathbb{T}_\pi}\! w(\theta)\,d\theta=1
\]
with $\sum_\theta w(\theta)=1$ in the discrete case and attach the rates to the channels, i.e.
\[
\gamma_\theta=\lambda\,w(\theta).
\]

Here $\lambda$ is an algorithmic Poisson intensity. It sets the expected number of analog kicks per unit time, while $w(\theta)$ specifies how kick angles are distributed. Thus $\gamma_\theta=\lambda w(\theta)$ is the rate assigned to the collapse operator $L_\theta=e^{i\theta P}$. We then choose $\lambda$ such that the averaged action of these kicks reproduces the physical Pauli-Lindblad rate $\gamma$ via the generator-matching condition $\lambda\,s=\gamma$, where $s=\mathbb E_w[\sin^2\theta]$. This is explained in detail in the following.

The corresponding Lindblad contribution is
\[
\sum_{\theta}\gamma_\theta\!\left(L_\theta\rho L_\theta^\dagger-\tfrac12\{L_\theta^\dagger L_\theta,\rho\}\right) = \sum_\theta \gamma_\theta\big(L_\theta\rho L_\theta^\dagger-\rho\big),
\]
since $L_\theta^\dagger L_\theta=\id$.
Using $L_\theta \rho L_\theta^\dagger = \cos^2\!\theta\,\rho+\sin^2\!\theta\,P\rho P + i\sin\theta\cos\theta\,[P,\rho]$ and the symmetry of $w$ to cancel the commutator on averaging, we get
\begin{equation}
\sum_\theta \gamma_\theta\big(L_\theta\rho L_\theta^\dagger-\rho\big)
= \Big(\sum_\theta \gamma_\theta \sin^2\!\theta\Big)\,(P\rho P-\rho)
= \lambda\,s\,(P\rho P-\rho),
\end{equation}
with $s:=\sum_\theta w(\theta)\sin^2\!\theta.$ Hence we enforce
\begin{equation}
\label{eq:match-generator}
\boxed{\ \lambda\,s=\gamma\ }\quad\Longrightarrow\quad \text{same generator as in \eqref{eq:analog-lindblad}}.
\end{equation}

In the MCWF, the jump hazard per unit time for a single channel is \\ $\sum_\theta\gamma_\theta\,\langle L_\theta^\dagger L_\theta\rangle = \sum_\theta \gamma_\theta = \lambda$, which is state-independent.
For a gate of duration $\Delta t$, the integrated hazard is $\Lambda=\lambda\Delta t$, and the probability of one
jump is $\mathbb{P}_{\rm jump} = 1-e^{-\Lambda}$, exactly as in standard quantum-jump sampling or in the TJM for Hamiltonian simulation with unitary Lindblad operators \cite{PlenioKnight, sander2025largescalestochasticsimulationopen, AnalogSimQuantinuum}.

We use two symmetric laws
\renewcommand{\theenumi}{\roman{enumi}}\begin{enumerate}
\item \emph{Two-point:} $\theta\in\{\pm\theta_0\}$ with equal weight, yielding $s=\sin^2\!\theta_0$.
\item \emph{Gaussian:} $\theta\sim\mathcal N(0,\sigma^2)$ (in practice, a symmetric discrete quadrature $\{\theta_k,w_k\}$), for which $s=\tfrac12(1-e^{-2\sigma^2})$ in the continuous limit (and $s=\sum_k w_k\sin^2\!\theta_k$ under discretization).
\end{enumerate}
Given a target per-channel physical rate $\gamma$, choose $\lambda=\gamma/s$ to satisfy \eqref{eq:match-generator}.
Small $\theta_0$ or $\sigma$ (i.e., operators close to identity) imply larger $\lambda$ and thus more frequent but milder jumps, which empirically reduces trajectory variance while preserving the exact mean dynamics \cite{AnalogSimQuantinuum}.

To keep the stochastic steps well-conditioned, we choose angles from a simple hazard policy. Let $\Gamma_{\mathcal G}:=\sum_{p\in\mathcal G}\gamma_p$ be the sum of Lindblad rates over an analog group $\mathcal G$, i.e. all processes use the same analog scheme.
Fix the analog parameters solely by exact generator matching,
\[
\lambda_p\,s=\gamma_p,
\]
and choose the angle law so as not to inflate the hazard.

\subsubsection{Two-point law} Take $\theta_0\in(0,\pi/2]$ with $s=\sin^2\theta_0$ and set
\[
\lambda_p=\gamma_p/s.
\]
To avoid any hazard increase under one-jump sampling, use the conservative choice
\[
\theta_0=\frac{\pi}{2}\quad\Rightarrow\quad s=1,\ \ \lambda_p=\gamma_p,
\]
such that the per-layer hazard equals the physical $\Gamma$.

\subsubsection{Gaussian law} If a Gaussian mixture is explicitly requested, pick an $s\in(0,\tfrac12]$ and set
\[
\sigma=\sqrt{-\tfrac12\ln(1-2s)},\quad \lambda_p=\gamma_p/s.
\]
For a discretized Gaussian quadrature $\{(\theta_k,w_k)\}_{k=1}^M$, compute $s=\sum_k w_k\sin^2\theta_k$ and again set $\lambda_p=\gamma_p/s$, which preserves the generator exactly. $s$ trades fewer but larger kicks (large $s$, small $\lambda_p$) against more frequent but milder kicks (small $s$, large $\lambda_p$). At fixed physical rates $\{\gamma_p\}$, the single-trajectory variance decreases monotonically as $s$ decreases.

The resulting collapse sets for a process $p$ are
\[
\bigl\{\ \sqrt{\lambda_p/2}\,e^{\pm i\theta_0 P}\ \bigr\} \quad\text{or}\quad \bigl\{\ \sqrt{\lambda_p\,w_k}\,e^{i\theta_k P}\ :\ k=1,\dots,M\ \bigr\},
\]
with $\sum L^\dagger_{p,\theta}L_{p,\theta}=\lambda_p\,\id$ per channel and total per-layer hazard $\Lambda=\sum_{p\in\mathcal G}\lambda_p$ (equal to $\Gamma$ in the default two-point choice).

Because the analog kicks are near identity (small $\theta_0$ or small $\sigma$) but occur more frequently (larger $\lambda$), single-trajectory fluctuations of typical observables are substantially reduced compared to standard Pauli-flip unravelings, while the ensemble mean remains unchanged by~\eqref{eq:match-generator} as in \cite{AnalogSimQuantinuum}.

\medskip
\subsubsection{Dissipative contraction for analog unraveling}
Assume all active single-site Pauli-Lindblad channels on site $\ell$ are indexed by $p\in S(\ell)$ with physical rates $\{\gamma_p\}_{p\in S(\ell)}$ and single-site Paulis $\{P^{(\ell)}_p\}_{p\in S(\ell)}$.
Each channel $p$ is unraveled by an analog unitary mixture with a symmetric angle law $w$ and
\begin{equation*}
L^{(\ell)}_{p,\theta} \;=\; \sqrt{\lambda_p\,w(\theta)}\,e^{i\theta P^{(\ell)}_p},
\lambda_p \;=\; \gamma_p/s, \quad s:=\mathbb E_w[\sin^2\theta]\in(0,1].
\end{equation*}
Note that we use the notation $L^{(\ell)}_{p}$ for a matrix in $\mathbb{C}^{2^L,2^L}$ expressed in a tensor product of $2 \times2$ matrices, whereas we used the notation $L^{[\ell]}_{j}$ for a single site jump matrix of size $2 \times2$ in \cref{sec:motivation}.
Since $e^{i\theta P}$ is unitary and $L^{(\ell)}_{p,\theta}{}^\dagger L^{(\ell)}_{p,\theta}=\lambda_p w(\theta)\,\id$, we have for each channel $p$
\[
\sum_{\theta} (L^{(\ell)}_{p,\theta})^\dagger L^{(\ell)}_{p,\theta}
= \Big(\sum_{\theta}\lambda_p w(\theta)\Big)\,\id
= \lambda_p\,\id.
\]
Summing over the local channel set $S(\ell)$ gives
\begin{equation*}
\begin{aligned}
&\sum_{p\in S(\ell)}\sum_{\theta} (L^{(\ell)}_{p,\theta})^\dagger L^{(\ell)}_{p,\theta}
= \Lambda_\ell\,\id,\\
&\Lambda_\ell:=\sum_{p\in S(\ell)}\lambda_p=\frac{1}{s}\sum_{p\in S(\ell)}\gamma_p = \frac{\Gamma_\ell}{s},
\end{aligned}
\end{equation*}
where $\Gamma_\ell:=\sum_{p\in S(\ell)}\gamma_p$ is the physical total rate on site $\ell$.
Finally, summing over sites yields

\begin{equation}
\begin{aligned}
&\sum_{\ell}\sum_{p\in S(\ell)}\sum_{\theta} (L^{(\ell)}_{p,\theta})^\dagger L^{(\ell)}_{p,\theta}
= \Big(\sum_{\ell}\Lambda_\ell\Big)\,\id \\
&= \frac{1}{s}\Big(\sum_{\ell}\Gamma_\ell\Big)\,\id
= \frac{\Gamma_{\mathrm{tot}}}{s}\,\id,
\end{aligned}
\end{equation}
with
\begin{equation}
\Gamma_{\mathrm{tot}}:=\sum_{\ell}\sum_{p\in S(\ell)}\gamma_p.
\end{equation}
Hence the non-Hermitian drift and the global dissipative factor over a window of duration $\delta t$ are
\begin{equation*}
H_{\mathrm{eff}} = H-\tfrac{i}{2}\sum_{\ell}\sum_{p\in S(\ell)}\sum_{\theta}(L^{(\ell)}_{p,\theta})^\dagger L^{(\ell)}_{p,\theta} = H - \tfrac{i}{2}\,\frac{\Gamma_{\mathrm{tot}}}{s}\,\id\\
\end{equation*}
and
\begin{equation*}
\begin{aligned}
&\mathcal{D}_{\text{analog}}[\delta t] = \exp\!\Big[-\tfrac{1}{2}\delta t\sum_{\ell}\sum_{p\in S(\ell)}\sum_{\theta} (L^{(\ell)}_{p,\theta})^\dagger L^{(\ell)}_{p,\theta}\Big]\\
&= \exp\!\Big(-\tfrac{\Gamma_{\mathrm{tot}}}{2s}\,\delta t\Big)\,\id.
\end{aligned}
\end{equation*}

Consequently, the trajectory norm survival probability is $e^{-(\Gamma_{\mathrm{tot}}/s)\,\delta t}$ (total hazard $\Gamma_{\mathrm{tot}}/s$), while the wave function is multiplied by the global scalar $e^{-(\Gamma_{\mathrm{tot}}/2s)\,\delta t}$, which cancels upon renormalization. This makes explicit the $S(\ell)$ summation over all single-site channels per site and the passage from physical rates $\gamma_p$ to algorithmic rates $\lambda_p=\gamma_p/s$ in the analog mixture.

\subsection{Projector sampling}
\label{sec:proj_unraveling}

For moderate-to-strong noise, the variance is often dominated by rare but large trajectory events, namely sharp flips that strongly disturb the state and inflate entanglement. In this regime we use a complementary unraveling based on projector jumps to keep the trajectory variance low.

Consider the multi-qubit Lindbladian with no Hamiltonian and a single Pauli-string channel
\begin{equation}
\label{eq:paulichannel-lindblad}
\frac{d\rho}{dt} \;=\; \gamma\big(P\rho P-\rho\big).
\end{equation}
This representation is not unique and one can make use of the unitary-mixing and identity shifts from \cite{Wiseman_2001} to reduce the variance of the trajectories.
The same dynamics can be simulated by taking the projector-jump unraveling with collapse operators
\begin{equation}
L_\pm \;:=\; \sqrt{\frac{\gamma}{2}}\,(\id\pm P)\;=\;\sqrt{2\gamma}\,\Pi_\pm, \quad \Pi_\pm \;:=\; \frac{\id\pm P}{2}.
\end{equation}

We now show the equivalence of the projector unraveling to the original dissipative Generator.
Because $P=P^\dagger$ and $P^2=\id$, the $\Pi_\pm$ are orthogonal projectors with
$\Pi_\pm^2=\Pi_\pm$, $\Pi_+\Pi_-=0$, and $\Pi_++\Pi_-=I$.
The corresponding Lindblad generator is given by
\begin{equation}
\label{eq:Lindblad-proj}
\begin{split}
  &\sum_{\pm} L_\pm \rho L_\pm^\dagger \;-\; \frac12\Big\{\sum_{\pm} L_\pm^\dagger L_\pm,\rho\Big\}\\
  &= 2\gamma\!\sum_{\pm}\Pi_\pm \rho \Pi_\pm \;-\; \gamma\{ \Pi_+ + \Pi_-, \rho\} \\
  &= 2\gamma\!\sum_{\pm}\Pi_\pm \rho \Pi_\pm \;-\; \gamma\{ \id, \rho\}.
\end{split}
\end{equation}
Using
\begin{align*}
&\sum_{\pm}\Pi_\pm \rho \Pi_\pm = \tfrac14\!\left[(\id{+}P)\rho(\id{+}P) + (\id{-}P)\rho(\id{-}P)\right]= \tfrac12\big(\rho + P\rho P\big),
\end{align*}
we obtain from \eqref{eq:Lindblad-proj} that
\begin{align*}
&\sum_{\pm} L_\pm \rho L_\pm^\dagger \;-\; \frac12\Big\{\sum_{\pm} L_\pm^\dagger L_\pm,\rho\Big\}\\
&= 2\gamma \cdot \tfrac12(\rho + P\rho P) \;-\; \gamma\{\id,\rho\} \nonumber\\
& = \gamma(\rho + P\rho P) - 2\gamma\rho \nonumber\\
&= \gamma\big(P\rho P - \rho\big),
\end{align*}
which coincides with~\eqref{eq:paulichannel-lindblad}.

The next theorem makes variance reduction under projector sampling quantitative.
In particular, it identifies a class of noise windows in which projector trajectories become effectively one-bit, leading to closed-form mean and variance expressions.

\begin{theorem}[Variance under projector-jump trajectories] \label{theorem:projector-variance}
Let the Lindbladian be a sum of Pauli-string channels with rates $\gamma_k$,
$\frac{d\rho}{dt} =\sum_k \gamma_k(P_k\rho P_k-\rho)$, and use the projector unraveling
$L_{k,\pm}=\sqrt{\gamma_k/2}\,(I\pm P_k)$.
Fix an observable $O$ and assume on a time interval $[0,t]$:
(i) no Hamiltonian/gate action; (ii) $\{O,P_k\}=0$ for all active channels $k$.
Let the initial pure state be an $O$-eigenstate with eigenvalue $+1$.
Then the single-trajectory estimator $X_t:=\langle O\rangle_{\mathrm{traj}}(t)\in\{1,0\}$ obeys
\begin{equation}\label{eq:proj-bernoulli-law}
\mathbb{P}[X_t=1]=e^{-2\Gamma_{\mathrm{anti}} t},
\quad \Gamma_{\mathrm{anti}}:=\sum_k\gamma_k.
\end{equation}
Therefore
\begin{equation}
  \mathbb{E}[X_t]=e^{-2\Gamma_{\mathrm{anti}} t}
\end{equation}
and
\begin{equation}
  \label{eq:proj-variance}
  \operatorname{Var}_{\mathrm{proj}}\!\big[\langle O\rangle_t\big] = \operatorname{Var}[X_t] = e^{-2\Gamma_{\mathrm{anti}} t}\!\left(1-e^{-2\Gamma_{\mathrm{anti}} t}\right).
\end{equation}
\end{theorem}
\begin{proof}
First we show, that the total jump rate is state-independent:
\begin{equation}
\begin{aligned}
&L_{k,\pm}^\dagger L_{k,\pm}
=\tfrac{\gamma_k}{2}(\id\pm P_k)^\dagger(\id\pm P_k)\\
&=\tfrac{\gamma_k}{2}\,(2\id \pm 2P_k)
= \gamma_k (\id\pm P_k),
\end{aligned}
\end{equation}
hence
\[
\sum_{\pm}L_{k,\pm}^\dagger L_{k,\pm}
= \gamma_k\big[(\id+P_k)+(\id-P_k)\big]=2\gamma_k\,\id.
\]
Summing over $k$ yields
\begin{equation}\label{eq:sumLdagL}
\sum_{k,\pm} L_{k,\pm}^\dagger L_{k,\pm} \;=\; 2\Gamma_{\mathrm{anti}}\id, \quad \Gamma_{\mathrm{anti}}:=\sum_k \gamma_k.
\end{equation}
In the quantum-jump picture, with $H=0$ the effective non-Hermitian drift is $H_{\rm eff}=-\tfrac{i}{2}\sum L^\dagger L=-i\,\Gamma_{\mathrm{anti}} \id$, so $\|\psi(t)\|^2=e^{-2\Gamma_{\mathrm{anti}} t}$ and this norm squared equals the no-jump probability on $[0,t]$, cf. \cite{PlenioKnight, Molmer:92}.

We continue with the absorbing property under anti-commuting projectors.
Fix $k$ with $\{O,P_k\}=0$ and write $\Pi_{k,\pm}=\tfrac12(\id\pm P_k)$. Then
\begin{equation}\label{eq:PiOPi-zero}
\begin{aligned}
&\Pi_{k,\pm} O \Pi_{k,\pm}
=\tfrac14 (\id\pm P_k) O (\id\pm P_k)\\
&=\tfrac14\big(O \pm OP_k \pm P_k O + P_k O P_k\big)
=0,
\end{aligned}
\end{equation}
because $P_k O = - O P_k$ and $P_k O P_k = - O$. Consequently, if a (normalized) jump of channel $k$ occurs,
\begin{equation}
\begin{aligned}
&\rho \;\mapsto\; \rho'_{k,\pm}
=\frac{\Pi_{k,\pm}\rho \Pi_{k,\pm}}{\Tr[\Pi_{k,\pm}\rho]} \\
&\Longrightarrow\quad
\langle O\rangle'_{k,\pm}=\Tr[O\rho'_{k,\pm}]
=\frac{\Tr[\Pi_{k,\pm} O \Pi_{k,\pm}\rho]}{\Tr[\Pi_{k,\pm}\rho]}=0.
\end{aligned}
\end{equation}
Thus, conditional on any first jump (from any anti-commuting channel), the trajectory’s estimator $X:=\langle O\rangle$ becomes equal to $0$ and, with no unitary dynamics, remains $0$ thereafter. We call this an absorbing property. (Contrast: commuting projectors can change coherences but preserve the unconditional mean of $O$.) \cite{PlenioKnight}

Now we take a look at the distribution of the estimator $X_t$ and the Bernoulli law.
Let $T$ denote the (random) time of the first jump. By (1), $T$ is exponentially distributed with rate $2\Gamma_{\mathrm{anti}}$, independent of the state. By (2), $X_t=1$ iff $T>t$ (no jump by time $t$), and $X_t=0$ otherwise. Therefore
\begin{equation}\label{eq:Bernoulli}
\begin{aligned}
&\mathbb P[X_t=1]=\mathbb P[T>t]=e^{-2\Gamma_{\mathrm{anti}} t},\\
&\mathbb P[X_t=0]=1-e^{-2\Gamma_{\mathrm{anti}} t},
\end{aligned}
\end{equation}
i.e., $X_t\in\{0,1\}$ is Bernoulli with parameter $e^{-2\Gamma_{\mathrm{anti}} t}$.
This proves \eqref{eq:proj-bernoulli-law}.
From \eqref{eq:Bernoulli},
\begin{equation}\label{eq:mean-var}
\mathbb E[X_t]=e^{-2\Gamma t},\quad
\operatorname{Var}[X_t]=e^{-2\Gamma t}\big(1-e^{-2\Gamma t}\big),
\end{equation}
which are Eqs.~\eqref{eq:proj-bernoulli-law}-\eqref{eq:proj-variance}.
Independently of trajectories, the Heisenberg equation for $O$ under the generator
$\mathcal L(\rho)=\sum_k \gamma_k(P_k\rho P_k-\rho)$ gives
\begin{equation}
\begin{aligned}
&\frac{d}{dt}\,\Tr[O\rho_t]
=\sum_k \gamma_k\,\Tr[O P_k\rho_t P_k - O\rho_t] \\
&=\sum_k \gamma_k\,\Tr[P_k O P_k \rho_t - O\rho_t]\\
&=\sum_k \gamma_k\,\Tr[(-O-O)\rho_t]
=-2\Gamma_{\mathrm{anti}}\,\Tr[O\rho_t],
\end{aligned}
\end{equation}
where we used $P_k O P_k=-O$ (anticommutation). With $\Tr[O\rho_0]=1$, the solution is $\Tr[O\rho_t]=e^{-2\Gamma_{\mathrm{anti}} t}$, which equals $\mathbb E[X_t]$ in \eqref{eq:mean-var}. Hence the trajectory estimator is unbiased and reproduces the master-equation mean exactly. \cite{Molmer:92,PlenioKnight}
\end{proof}

Since the conditions of the theorem may sound like this only applies to very special cases, we derive an upper and lower bound of the trajectory variance of projector unraveling in the following section.

\subsubsection{Circuit relevance via absorbing windows}

We use the term absorbing window to describe a gate-aligned noise window with an absorbing-state behavior for a chosen readout observable $O$ under projector sampling.
Concretely, consider one noise application window (e.g.\ ``after gate $g$'').
Let $\{(\gamma_k,P_k)\}$ be the active Pauli-string channels in that window and assume that $O$ anticommutes with all of them, i.e., $\{O,P_k\}=0$ for every active $k$ (equivalently $\Gamma_{\rm comm}=0$).
Under the projector unraveling $L_{k,\pm}=\sqrt{\gamma_k/2}\,(I\pm P_k)$, the first jump projects the state into a $P_k$-eigenspace and forces $\langle O\rangle$ to $0$ via $\Pi_{k,\pm} O \Pi_{k,\pm}=0$.
Since there is no unitary evolution inside the window, $\langle O\rangle$ remains $0$ for the rest of the window.
Thus, within such an absorbing window, the trajectory estimator is binary: it equals its input value if and only if no jump occurred, and it is $0$ otherwise, yielding the Bernoulli mean/variance law of Theorem~\ref{theorem:projector-variance}.

Consider one two-qubit gate application followed by a noise sampling process and let $O$ be the
observable we read out in this circuit frame right after the gate.
Let $m:=\langle O\rangle_{\rm in}$ be its value before the noise step and $\langle O\rangle_{\rm out}$ after the noise step. Partition the active local channels into those that anticommute with $O$ and those that commute with $O$, and set
\begin{equation}
\begin{aligned}
&\Gamma_{\rm anti}:=\sum_{\{O,P_k\}=0}\gamma_k, \quad \Gamma_{\rm comm}:=\sum_{[O,P_k]=0}\gamma_k,\\
&\Gamma_{\mathrm{tot}}:=\Gamma_{\rm anti}+\Gamma_{\rm comm}.
\end{aligned}
\end{equation}

Note that two arbitrary Pauli strings always either commute or anticommute \cite{Aaronson_2004}.
Under the projector unraveling $L_{k,\pm}=\sqrt{\gamma_k/2}\,(\id\pm P_k)$ the total jump rate in the step is state-independent and equals $2\Gamma_{\mathrm{tot}}$ (we take $\Delta t=1$). The jump probability is $\mathbb{P}_{\rm jump}=1-e^{-2\Gamma_{\mathrm{tot}}}$. If a jump occurs, it is anti-commuting with probability $\Gamma_{\rm anti}/\Gamma_{\mathrm{tot}}$ and commuting with probability $\Gamma_{\rm comm}/\Gamma_{\mathrm{tot}}$. Thus we can partition the outcomes into three mutually exclusive events

\begin{equation}
\begin{aligned}
&E_0=\{\text{no jump}\},\\
&E_{\rm c}=\{\text{commuting jump(s)}\},\\
&E_{\rm a}=\{\text{anti-commuting jump(s)}\},
\end{aligned}
\end{equation}
with $\mathbb{P}[E_0]=e^{-2\Gamma_{\mathrm{tot}}}$.
Hence $\mathbb{P}[E_{\rm c}\cup E_{\rm a}]=1-e^{-2\Gamma_{\mathrm{tot}}}$,
$\mathbb{P}[E_{\rm c}]=(1-e^{-2\Gamma_{\mathrm{tot}}}) \Gamma_{\rm comm}/\Gamma_{\mathrm{tot}}$ and $\mathbb{P}[E_{\rm a}]=(1-e^{-2\Gamma_{\mathrm{tot}}}) \Gamma_{\rm anti}/\Gamma_{\mathrm{tot}}$.
With this we can now exactly calculate the per-step mean and per-step variance.

We start by calculating the per-step mean.
Commuting jumps do not change the per-step mean of $\langle O\rangle_{\text{out}}$,
whereas anti-commuting jumps force $\langle O\rangle_{\text{out}}\!\mapsto\!0$ in that step. Hence
\begin{align}
\label{eq:mixed-mean}
&\mathbb E[\langle O\rangle_{\rm out}] = m\Big[1-\mathbb{P}[E_{\rm a}]\Big]
= m\!\left[1-\frac{\Gamma_{\rm anti}}{\Gamma}\bigl(1-e^{-2\Gamma}\bigr)\right].
\end{align}

Now we can use that result to calculate the variance.
By construction of the projector channels $L_{k,\pm}\propto(\id\pm P_k)$, we have:
\[
\langle O\rangle_{\rm out}=
\begin{cases}
  m, & \text{on } E_0,\\
  \text{random value with }\mathbb E[\langle O\rangle_{\rm out}\!\mid\!E_{\rm c}]=m, & \text{on } E_{\rm c},\\
  0, & \text{on } E_{\rm a},
\end{cases}
\]
since $\Pi_{\pm}O\Pi_{\pm}=0$ for $\{O,P_k\}=0$.
Set
\begin{align}
&p:=\mathbb{P}[\langle O\rangle_{\rm out}=m]=\mathbb{P}[E_0]+\mathbb{P}[E_{\rm c}]
\\&=e^{-2\Gamma}+(1-e^{-2\Gamma})\frac{\Gamma_{\rm comm}}{\Gamma}.
\end{align}
Then $\langle O\rangle_{\rm out}$ is a mixture of the two point masses $\{m,0\}$ with weights $p$ and $1-p$, plus the within-commuting spread on $E_{\rm c}$.
Applying the law of total variance,
\[
\operatorname{Var}[\langle O\rangle_{\rm out}]
= \underbrace{\operatorname{Var}\big(\mathbb E[\langle O\rangle_{\rm out}\!\mid\!E]\big)}_{\text{between-events spread}}
+\underbrace{\mathbb E\big(\operatorname{Var}[\langle O\rangle_{\rm out}\!\mid\!E]\big)}_{\text{within-events spread}},
\]
where $E \in \{E_{\rm c}, E_{\rm a}, E_{\rm 0}\}$, gives
\begin{equation}
\label{eq:variance-decomp-detailed}
\operatorname{Var}[\langle O\rangle_{\rm out}]
=\underbrace{m^2\,p(1-p)}_{\text{between } \{m,0\}\text{ from }E_0\text{ vs }E_{\rm a}}
+\underbrace{\mathbb{P}[E_{\rm c}]\,\sigma_{\rm comm}^2}_{\text{within }E_{\rm c}},
\end{equation}
where
\[
\sigma_{\rm comm}^2\;:=\;\operatorname{Var}[\langle O\rangle_{\rm out}\mid E_{\rm c}]
\]
is the conditional variance generated by the state-dependent $\pm$ branches of the commuting projector(s).
Because $\operatorname{spec}(O)=\{\pm1\}$ and $\mathbb E[\langle O\rangle_{\rm out}\mid E_{\rm c}]=m$, we have the universal bound
\[
0\ \le\ \sigma_{\rm comm}^2\ \le\ 1-m^2,
\]
with the upper limit attained when the commuting jump(s) fully resolve $O$ to its eigenvalues $\pm1$ in that step.
Using $\mathbb{P}[E_{\rm c}]=(1-e^{-2\Gamma})\Gamma_{\rm comm}/\Gamma$, \eqref{eq:variance-decomp-detailed} yields the sharp sandwich

\begin{equation}
\begin{aligned}
\label{eq:variance-bounds-detailed}
&m^2\,p(1-p)\ \le\ \operatorname{Var}[\langle O\rangle_{\rm out}] \\
&\le\ m^2\,p(1-p)\;+\;(1-e^{-2\Gamma})\,\frac{\Gamma_{\rm comm}}{\Gamma}\,(1-m^2),
\end{aligned}
\end{equation}

and both bounds are attainable. The lower bound by an absorbing window with $\Gamma_{\rm comm}=0$, the upper bound by full polarization of the commuting projectors.
The absorbing-window condition holds exactly for $|0\ldots0\rangle$ inputs and $Z$-type readout observables whenever the back-propagated observable $\tilde O$ remains a $Z$-string and all active SPLM channels anticommute with it (i.e., $\Gamma_{\rm comm}=0$ on that window). 
This occurs, for example, in Clifford segments (CNOT/CZ/SWAP/H/S), CZ brickwork with $Z$ measurements, and Ising layers at Clifford angles such as $R_{ZZ}(\pi/2)$. 
For generic non-Clifford layers $\tilde O$ becomes a linear combination of Paulis, so \eqref{eq:proj-variance} may fail, although projector sampling typically still reduces variance.

\subsection{Long-range noise MPO construction for unravelings}
\label{sec:longrange-mpo}

Let $P$ be a Pauli string supported on two non-adjacent sites $i<j$, e.g.\ $P=\sigma_i\otimes\tau_j$ with $\sigma,\tau\in\{X,Y,Z\}$.
In both the projector and the analog unravelings, each collapse operator reduces (up to a global scalar) to a linear combination
\[
a\,\id\;+\; b\,P,\quad a,b\in\mathbb{C}.
\]
Such an operator admits an exact matrix product operator (MPO) with bond dimension $D=2$, independent of the separation $j-i$.
Writing the local $2\times2$ identity as $\id$, the MPO tensors (blocks are physical operators) are
\[
\begin{aligned}
  &\text{for } \ell<i:\quad W_\ell \;=\; \id, \\
  &\text{site } i:\quad
  W_i \;=\; \bigl[\, \id\;\; \sigma \,\bigr], \\
  &\text{for } i<\ell<j:\quad
  W_\ell \;=\;
  \begin{bmatrix} \id& 0 \\[2pt] 0 & \id
  \end{bmatrix}, \\
  &\text{site } j:\quad
  W_j \;=\;
  \begin{bmatrix} a\,\id\\[2pt] b\,\tau
  \end{bmatrix}, \\
  &\text{for } \ell>j:\quad W_\ell \;=\; \id,
\end{aligned}
\]
so that contracting the virtual (bond-2) indices from left to right yields
$
\cdots \id\cdot [\,\id\ \sigma\,]\cdot \mathrm{diag}(\id,\id)\cdots \mathrm{diag}(\id,\id)\cdot
\begin{bmatrix} a\id\\ b\tau
\end{bmatrix}\cdot \id\cdots = a\,\id+ b\,(\sigma_i\otimes\tau_j)=a\id+bP$.
Only the coefficients differ for analog and projector unraveling:

\begin{itemize}
\item Projector unraveling:
  \begin{align*}
    &\Pi_\pm=\tfrac12(\id\pm P)
    \ \Longleftrightarrow\ (a,b)=\bigl(\tfrac12,\ \pm\tfrac12\bigr),\\
    &L_\pm=\sqrt{\tfrac{\gamma}{2}}\,(\id\pm P)
    \ \Longleftrightarrow\ (a,b)=(1,\ \pm 1)\, \\
    &\text{(overall factor }\sqrt{\tfrac{\gamma}{2}})
  \end{align*}
\item Analog unraveling:
  \begin{align*}
    &U_{\theta}=e^{i\theta P}=\cos\theta\,\id+ i\sin\theta\,P
    \ \Longleftrightarrow\ (a,b)=\bigl(\cos\theta,\ i\sin\theta\bigr),\\
    &L_{\theta}=\sqrt{\lambda\,w(\theta)}\,U_{\theta}
    \ \Longleftrightarrow\ (a,b)=\bigl(\cos\theta,\ i\sin\theta\bigr)\, \\
    &\text{(overall factor }\sqrt{\lambda\,w(\theta)}).
  \end{align*}
\end{itemize}

Applying any long-range projector or analog jump therefore amounts to a single bond-2 MPO contraction across the sites $i,\ldots,j$, after which standard truncation can be used.

\section{Numerical Experiments}

In this section we show the scalability of cTJM due to variance reduction and limited bond dimension growth under several noise regimes with 25 and 127 qubit examples. Furthermore we show that the variance of a 2-qubit bitflip channel aligns with the theory for all unravelings.

\subsection{Numerical variance check (two-qubit bit-flip SPLM)}
\label{sec:variance-check}

We validate the variance predictions of Sec.~\ref{sec:analog-unraveling} and the projector analysis by a controlled two-qubit test in which the only dynamics is noise.
We use the two-qubit sparse Pauli-Lindblad model with equal rates $\gamma$ on the strings
$\{\,X\!\otimes\!\idtwo,\ \idtwo\!\otimes\!X,\ X\!\otimes\!X\,\}$, initialize $\ket{00}$, and apply a two-qubit identity layer between noise layers. We record the single-trajectory estimators $\langle Z_i\rangle_{\mathrm{traj}}(t)$ after every layer and average over trajectories.

In this benchmark the total physical rate is $\Gamma_{\mathrm{tot}} := \sum_k \gamma_k = 3\gamma,
$ but only $\Gamma_{\mathrm{anti}}(Z_i) := \sum_{\{Z_i,P_k\}=0}\gamma_k = 2\gamma$ affects $Z_i$.

For the local observables $Z_i$, exactly two Pauli strings anticommute with $Z_i$, namely the local $X_i$ and the correlated $X_0X_1$. Denoting their total rate by $\Gamma_{\mathrm{anti}}(Z_i)=\gamma_{X_i}+\gamma_{X_0X_1}=2\gamma$, one finds for the ensemble mean
\begin{equation}
\mathbb{E}\big[\langle Z_i\rangle_t\big] \;=\; e^{-2 \Gamma_{\mathrm{anti}}(Z_i) t} \;=\; e^{-4\gamma t}.
\end{equation}
The across-trajectory variances differ by unraveling:
\begin{equation}
\label{eq:var-standard}
\operatorname{Var}_{\text{std}}\!\big[\langle Z_i\rangle_t\big] = 1 - e^{-4 \Gamma_{\mathrm{anti}}(Z_i) t} \;=\; 1 - e^{-8\gamma t},
\end{equation}
\begin{equation}
\operatorname{Var}_{\text{proj}}\!\big[\langle Z_i\rangle_t\big] = e^{-4\gamma t}\!\left(1 - e^{-4\gamma t}\right),
\end{equation}
where the first line reflects a telegraph process under unitary $X$-jumps, and the second follows from the absorbing ``first anticommute projector jump'' with rate $2\Gamma_{\mathrm{anti}}(Z_i)=4\gamma$.
For the analog two-point and Gaussian schemes, an exact transient is not needed here; in this symmetric two-qubit IX/XI/XX case the stationary across-trajectory variance is
\begin{equation}
\lim_{t\to\infty}\operatorname{Var}_{\text{analog}}\!\big[\langle Z_i\rangle_t\big] = \tfrac14.
\end{equation}

For a detailed derivation of the variance see \cref{appendix:analog-variance-2pt-gauss}.
Because the identity-layer test yields independent, stationary noise steps and a fixed
anticommutation structure for $Z_i$, the one-step moment map for $\langle Z_i\rangle$
is identical at every layer. Iterating this same map over $t$ layers therefore produces
exactly the time-$t$ trajectory mean and variance, so our “per-step” formulas
coincide with the full-time curves in this benchmark. In general noisy circuits, the same
exact composition holds on absorbing windows (where the conjugated readout and the
active channels preserve the anticommute pattern); outside such windows, the per-step
formulas serve as local bounds or approximations.

We run cTJM with $N_{\text{traj}}=2000$ trajectories, no truncation effects, and identity gates only.
For analog runs we use a fine Gaussian discretization so that the discrete mixture closely matches the continuous law.

Figure~\ref{fig:variance-2q} (top) plots the measured variances for $Z_0$ and $Z_1$ against the theoretical curves above (dashed).
The projector unraveling follows $e^{-4\gamma t}(1-e^{-4\gamma t})$ and rapidly decays to (numerically) zero.
The standard unraveling approaches unit variance as $1-e^{-8\gamma t}$.
Both analog schemes settle to the predicted plateau $1/4$.
Figure~\ref{fig:variance-2q} (bottom) shows that the ensemble mean $\mathbb{E}[\langle Z_i\rangle_t]$ coincides for all unravelings and matches the density-matrix baseline $e^{-4\gamma t}$.
Overall, projector sampling attains the lowest variance for $\gamma = 0.1, 0.01$ while preserving the mean dynamics, in agreement with our theory.

\begin{figure*}[!t]
\centering
\includegraphics[width=\textwidth]{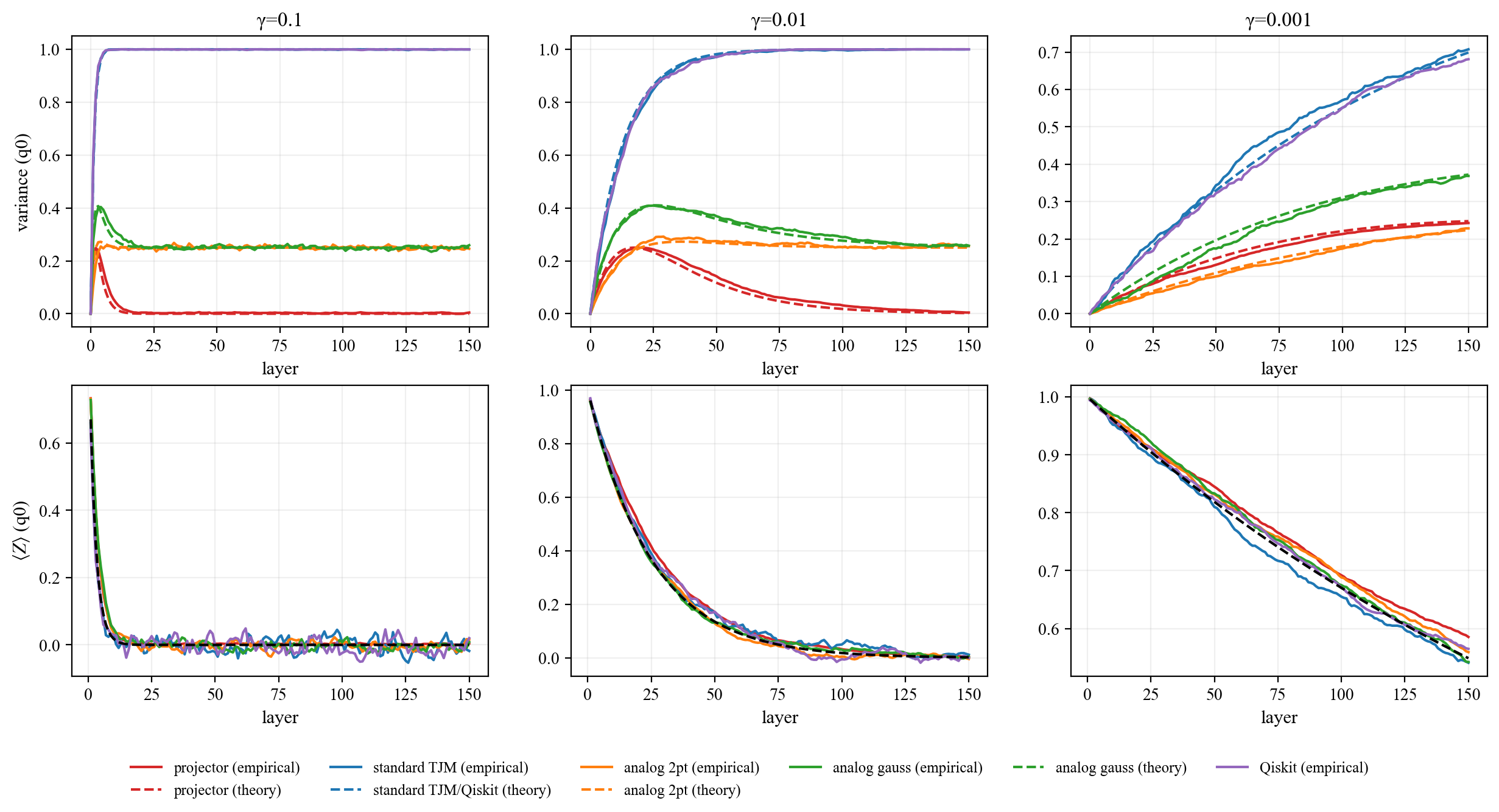}
\caption{\textbf{Numerical variance check for two-qubit bitflip noise:} Empirical and theoretical variances (top) and mean Z expectation values (bottom) for the two-qubit bitflip noise benchmark at rates $\gamma\in\{10^{-3},10^{-2},10^{-1}\}$. Each curve shows the single-qubit observable $\langle Z_0\rangle$ averaged over $N_{\text{traj}}=2000$ trajectories. Solid lines are empirical results from Qiskit (purple), standard TJM (blue), analog Gaussian (green), analog two-point (orange), and projector unraveling (red). Dashed lines indicate analytic predictions: \cref{eq:var-2pt},\cref{eq:var-gauss}, \cref{eq:var-standard}  for analog schemes and \cref{eq:proj-variance} for the projector unraveling. The projector trajectories saturate at the lowest variance for high noise, while the analog unravelings converge to the predicted stationary value $\operatorname{Var}\to1/4$ for weak noise. All unravelings reproduce the same mean decay $\mathbb E[\langle Z\rangle_t]=e^{-4\gamma t}$, confirming generator-matching across noise strengths.}
\label{fig:variance-2q}
\end{figure*}

\subsection{Noisy XY chain (25 sites)}
\label{sec:xy-quench-bitflip}
As a first many-body benchmark at circuit-relevant scale we consider 20 Trotter steps of a global quench in the open-boundary spin-$\tfrac12$ XY chain,
\begin{equation}
H_{\mathrm{XY}}
\;=\;
\sum_{\ell=1}^{N-1} \Big( X_\ell X_{\ell+1} + Y_\ell Y_{\ell+1} \Big),
\end{equation}
simulated here by a Trotterized circuit on $N=25$ qubits with timestep $\delta t=0.1$ (total simulated time $T=2.0$).
The XY chain is a standard integrable testbed for nonequilibrium transport and entanglement growth in one dimension; see, e.g., \cite{XY_PhysRevA.3.786}.

We start from the product state $|0\rangle^{\otimes N}$ and apply Pauli-$X$ on every fourth qubit (1-indexed sites $4,8,12,16,20,24$), i.e.
\begin{equation}
\ket{\psi_0}
\;=\;
\Big(\prod_{\ell\in\{4,8,12,16,20,24\}} X_\ell\Big)\,|0\rangle^{\otimes 25},
\end{equation}
which yields the periodic pattern $\ket{0001\,0001\,0001\cdots}$ and thus nontrivial initial magnetization gradients.

After each two-qubit gate we apply a homogeneous sparse Pauli-Lindblad noise consisting of nearest-neighbor correlated flips $X_\ell X_{\ell+1}$ on all bonds, each with identical rate parameter $\gamma\in\{10^{-3},10^{-2},10^{-1}\}$:
\begin{equation}
\label{eq:xy_noise_layer}
\mathcal{L}_{X}(\rho) =\gamma\sum_{\ell=1}^{N-1}\big(X_\ell X_{\ell+1}\rho X_\ell X_{\ell+1}-\rho\big).
\end{equation}

We compare five trajectory-based simulators:
Qiskit Aer’s MPS reference and the four cTJM unravelings discussed in \cref{sec:cTJM}.
We simulate $N_{\mathrm{traj}}=200$ independent trajectories per method,
cap the bond dimension at $\chi_{\max}=128$,
and use an SVD truncation threshold of $10^{-16}$.

We record local magnetizations $\langle Z_\ell\rangle$.
To quantify sampling efficiency we track
(i) the across-trajectory variance of the single-trajectory estimator(s),
and (ii) the average bond dimension (averaged over trajectories at fixed step, with $\chi\le\chi_{\max}$).

Figure~\ref{fig:large_tau_compare} reports three noise strengths.
The top panels show representative local expectations $\langle Z_\ell\rangle$ for sites $\ell=1$ (solid),
$\ell=12$ (dashed, initially flipped to $\langle Z_{12}\rangle=-1$), and $\ell=25$ (dotted).
Across all $\gamma$, all cTJM unravelings remain unbiased and track the \texttt{Qiskit} MPS reference in the mean.

At weak noise ($\gamma=10^{-3}$) the dynamics is close to coherent XY transport:
the local $\langle Z_\ell\rangle$ traces oscillate and dephase slowly, while the trajectory variance remains small for all methods.
The projector unraveling yields the smallest variance across the full depth,
the analog schemes reduce variance relative to standard Pauli jumps,
and the standard unraveling exhibits the largest fluctuations.

At intermediate noise ($\gamma=10^{-2}$) the projector unraveling again provides the strongest variance reduction and, crucially, a visibly slower growth in bond dimension, staying below the hard cap $\chi_{\max}=128$ up to depth $20$,
whereas the other unravelings saturate near $\chi_{\max}$ around step $\approx 14$.

In the strongly dissipative regime ($\gamma=10^{-1}$) the projector sampling rapidly suppresses trajectory-to-trajectory spread essentially to zero after the initial transient,
consistent with its absorbing-window behavior discussed in \cref{sec:cTJM}.
Simultaneously, the bottom panel demonstrates the most dramatic computational advantage:
projector trajectories remain extremely low-entangled (average bond dimension $\chi\sim\mathcal{O}(1)$),
while standard and analog unravelings continue to generate substantially larger bond dimensions.
This benchmark therefore highlights the practical regime separation:
analog mixtures are beneficial in low-noise settings where small, frequent unitary kicks reduce fluctuations without collapsing the state,
whereas projector jumps are strongly advantageous in moderate-to-high noise, where they simultaneously reduce Monte Carlo variance
and suppress MPS entanglement growth, enabling deeper simulation within fixed $\chi_{\max}$.

\begin{figure*}[!t]
\centering

\includegraphics[width=1.0\linewidth]{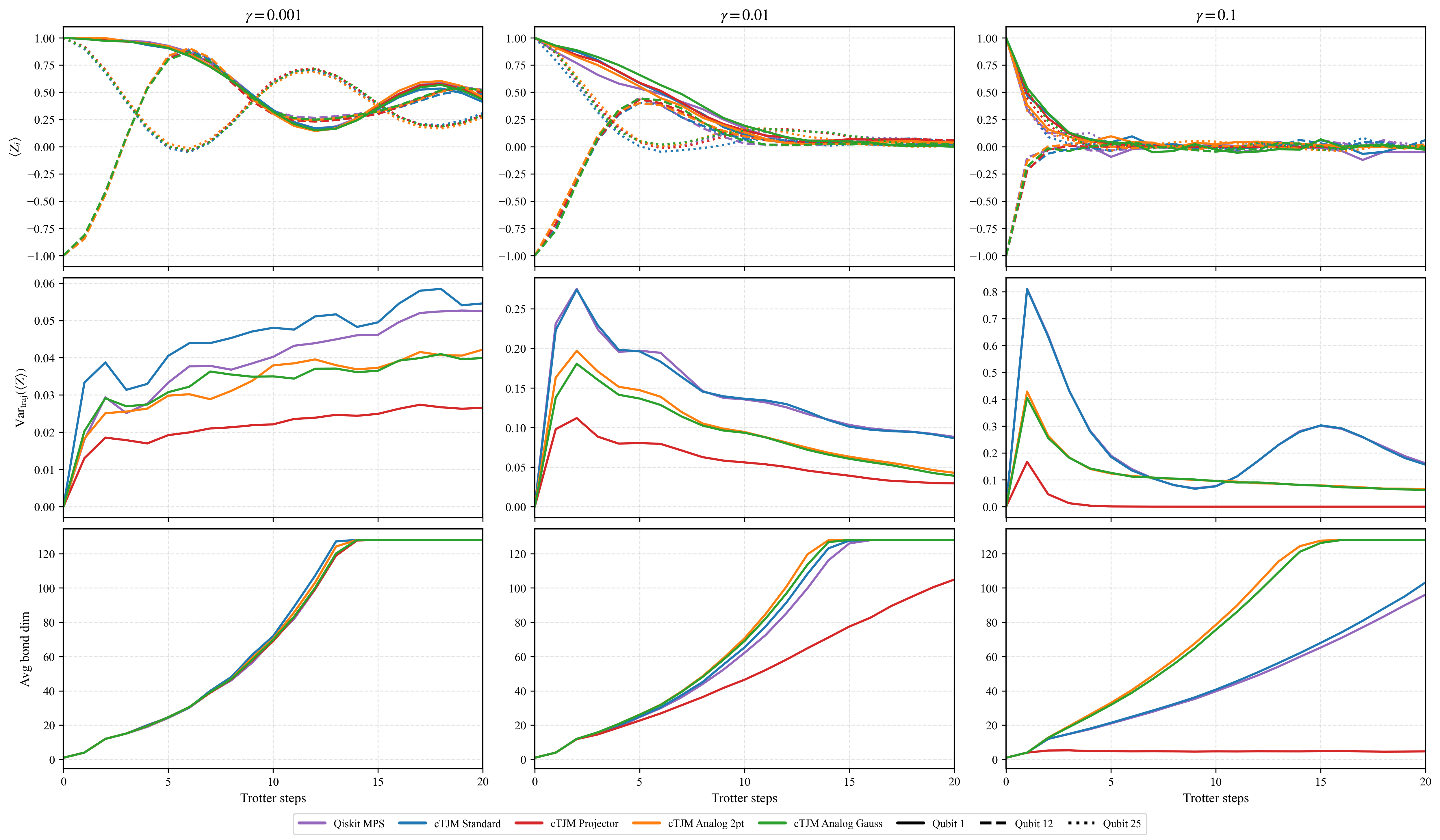}

\caption{\textbf{Noisy XY quench with correlated bit-flip SPLM noise on a chain of $N=25$ qubits.}
  We simulate $20$ first-order Trotter steps with $\delta t=0.1$ ($T=2.0$) starting from $\ket{0001\,0001\,\cdots}$ (flips on sites $4,8,12,16,20,24$).
  After each two-qubit layer we apply the homogeneous Lindbladian channel generated by Eq.~\eqref{eq:xy_noise_layer}, with $\gamma\in\{10^{-1},10^{-2},10^{-3}\}$.
  Each panel shows (top) the trajectory mean of $\langle Z_\ell\rangle$ for $\ell=1$ (solid), $\ell=12$ (dashed), and $\ell=25$ (dotted), (middle) the across-trajectory variance of the single-trajectory estimator, and (bottom) the mean MPS bond dimension per step.
We compare Qiskit Aer’s MPS baseline to four cTJM unravelings using $N_{\mathrm{traj}}=200$, truncation threshold $10^{-16}$, and hard cap $\chi_{\max}=128$.}

\label{fig:large_tau_compare}
\end{figure*}

\subsection{Noisy IBM kicked Ising (127 sites)}
\label{sec:127-qubits}

To demonstrate the scalability of cTJM in a hardware-motivated setting, we simulate the 127-qubit kicked-Ising benchmark circuit used on IBM Eagle in Ref.~\cite{kim2023evidence}. The target dynamics is the transverse-field Ising Hamiltonian on the IBM Eagle heavy-hex connectivity
\begin{equation}
H = -J \sum_{\langle i,j\rangle} Z_i Z_j + h \sum_i X_i ,
\end{equation}
whose real-time evolution is implemented via a first-order Trotter step of the form

\begin{equation}
U(\theta_h)=
\left(\prod_{\langle i,j\rangle}\exp\!\left(i\frac{\pi}{4}\,Z_i Z_{j}\right)\right)
\left(\prod_{i}\exp\!\left(-i\frac{\theta_h}{2}\,X_i\right)\right),
\end{equation}

with gate angles $\theta_J=-2J\delta t$ and $\theta_h=2h\delta t$ and Trotter step size $\delta t$. Following ~\cite{kim2023evidence} we fix $\theta_J=-\pi/2$, throughout we use $\delta t=0.1$, hence $J=5\pi/2$, and vary
\begin{equation}
\theta_h \in \left\{0,\ \frac{\pi}{8}, \frac{\pi}{2}\right\}
\quad\Leftrightarrow\quad
h=\frac{\theta_h}{2\delta t}.
\end{equation}
We follow IBM's hardware qubit labels $0,\dots,126$ and report the local observable $\langle Z_{106}\rangle$~\cite{kim2023evidence}.

Beyond homogeneous local noise, cTJM supports gate-dependent 
two-qubit CPTP noise channels, including long-range (non-nearest-neighbor) two-qubit channels used to model correlated errors/crosstalk in a device-calibrated manner.
In this benchmark we scale an overall depolarizing noise strength by a factor $\gamma\in\{10^{-3},10^{-2},10^{-1}\}$ and apply the corresponding two-qubit depolarizing noise channel after each entangling gate (which can be extended analogously for single-qubit layers or other noise models if desired).

For $\theta_h=0$, the $R_X$ layer is identity and the $R_{ZZ}(-\pi/2)$ layer acts as a global phase on $\ket{0}^{\otimes n}$; the circuit is therefore identity-equivalent on the chosen initial state and serves as a baseline isolating pure noise effects at $n=127$.
For $\theta_h=\pi/2$, both $R_X(\pi/2)$ and $R_{ZZ}(-\pi/2)$ are Clifford and the circuit exhibits nontrivial entangling dynamics while remaining Clifford.
For $\theta_h=\pi/8$, the circuit becomes non-Clifford; exact stabilizer-tableau simulation no longer applies, whereas Pauli-propagation methods remain applicable but generally incur Pauli branching and require truncation to remain efficient~\cite{Aaronson_2004, angrisani2025simulatingquantumcircuitsarbitrary}.

All $100$ trajectories are simulated using soft SVD truncation with threshold $10^{-16}$ (near machine precision) and a hard truncation at $\chi_{\text{max}}=128$, so the observed bond dimension growth is dictated by the trajectory entanglement/noise tradeoff rather than an aggressive rank cap up until $\chi_{\text{max}}$. We initialize in $\ket{0}^{\otimes 127}$ as in \cite{kim2023evidence}.

Figure~\ref{fig:127q_plot} compares projector (solid) and standard (dashed) unravelings for nine noisy 127-qubit kicked-Ising settings (five Trotter steps) with two-qubit depolarizing noise of strength $\gamma\in\{10^{-1},10^{-2},10^{-3}\}$ applied after each entangling gate on the hardware connectivity.
As expected for depolarizing noise, increasing $\gamma$ systematically drives the local magnetization $\langle Z_{106}\rangle$ toward zero. Depolarizing channels shrink Bloch components uniformly, hence suppress local Pauli expectations~\cite{Xiong_2020, BRE02}. Both unravelings generate the same Lindblad evolution and are therefore unbiased; accordingly, the trajectory means (top row) agree closely between solid and dashed curves within finite-$N_{\mathrm{traj}}$ sampling fluctuations. The main practical difference is the estimator variance and, to a lesser extent, the trajectory entanglement reflected in the MPS bond dimension.

The three columns illustrate distinct dynamical regimes.
For $\theta_h=0$ the unitary is identity-equivalent on $\ket{0}^{\otimes 127}$, so any decay of $\langle Z_{106}\rangle$ is purely noise-induced, providing a baseline that isolates the action of the two-qubit noise channel.
For $\theta_h=\pi/2$, the mean rapidly approaches $\langle Z_{106}\rangle\approx 0$ largely independently of $\gamma$, indicating that unitary scrambling dominates the decay of local $Z$-magnetization after the first step(s), while noise primarily affects fluctuations.
For the non-Clifford setting $\theta_h=\pi/8$, the mean remains strongly noise-dependent and the simulation cost is entanglement-limited: only this case reaches the cap $\chi_{\max}=128$ at low noise $\gamma=10^{-3}$.

The middle row reports the empirical trajectory variance of $\langle Z_{106}\rangle$ for both unravelings. For $\theta_h \in \{0, \pi/8\}$ the projector unraveling keeps the variance uniformly bounded (peaking around $\sim 0.4$), whereas the standard unraveling can exhibit substantially larger fluctuations, most prominently in the baseline $\theta_h=0$ where the variance quickly approaches order one at stronger noise. Since the sampling error satisfies $\mathrm{SE}=\sqrt{\mathrm{Var}/N_{\mathrm{traj}}}$, the reduced variance directly translates into tighter error bars for fixed trajectory budget~\cite{Molmer:92}. We observed that also in longer runs of up to 20 Trotter steps the projector-unraveling variance never permanently exceeded $0.4$.

For the strongly scrambling Clifford case $\theta_h=\pi/2$, the projector variance can be higher because commuting projector jumps induce measurement backaction and add a nonzero within-event contribution $\mathbb{P}[E_{\rm c}]\sigma_{\rm comm}^2$ (cf. Eq.~\eqref{eq:variance-decomp-detailed}), whereas commuting unitary Pauli jumps in the standard unraveling leave $\langle Z_{106}\rangle$ unchanged.

Finally, the bottom row visualizes the average MPS bond dimension for projector (solid) and standard (dashed) unravelings, illustrating the trade-off between coherent entanglement growth (larger bond dimension at smaller $\gamma$) and noise-induced trajectory entanglement (visible even at $\theta_h=0$). Differences between unravelings are regime-dependent but remain secondary to the overall entanglement-noise trends; in particular, in the non-Clifford case $\theta_h=\pi/8$ the simulation becomes entanglement-limited, reaching the cap $\chi_{\max}=128$ at $\gamma=10^{-3}$ for both unravelings, and already at $\gamma\in\{10^{-2},10^{-1}\}$ for the standard unraveling.

\begin{figure*}[!t]
\centering
\includegraphics[width=1\linewidth]{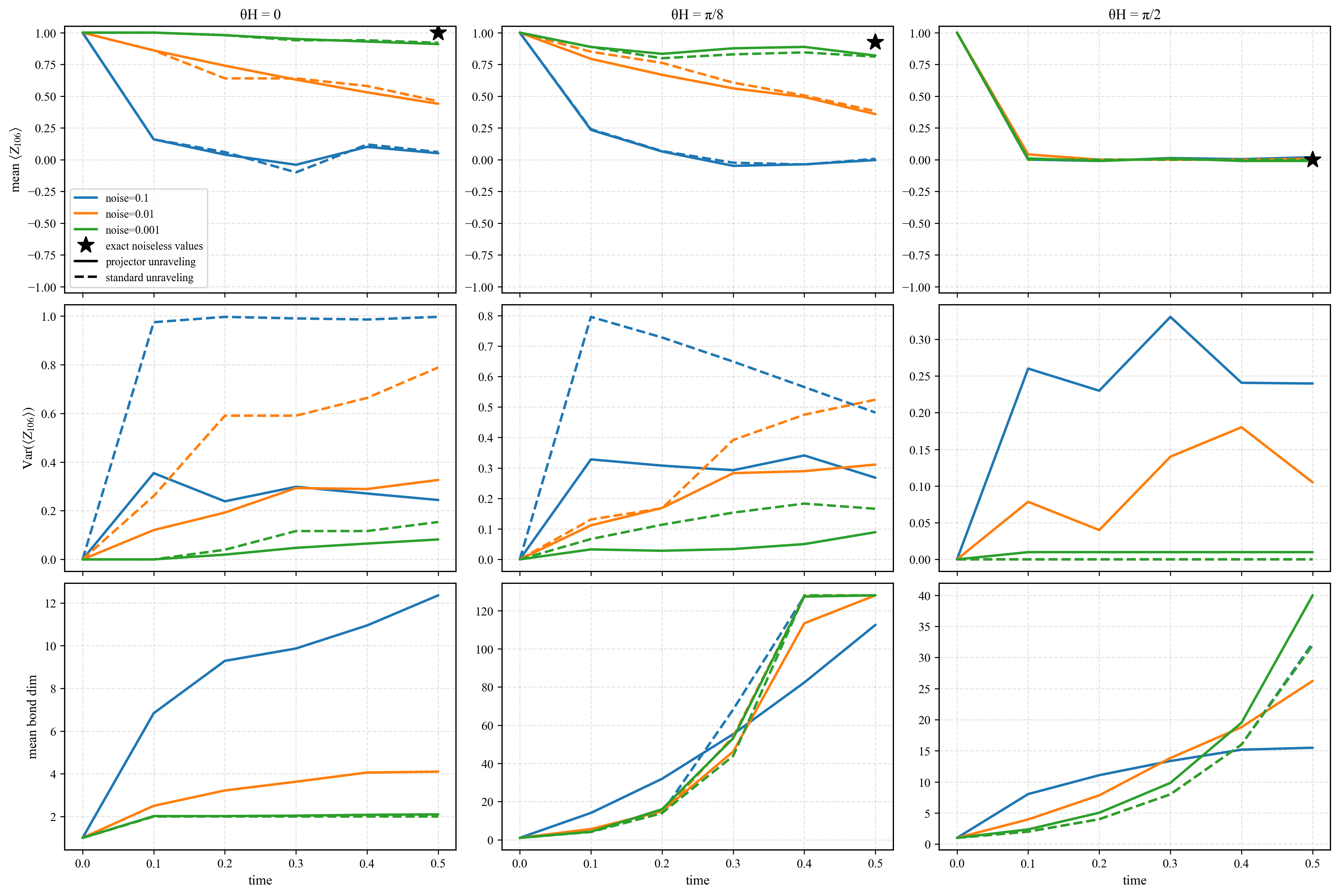}
\caption{\textbf{Noisy 127-qubit kicked-Ising dynamics on IBM Eagle connectivity.}
  We simulate five Trotter steps ($\delta t=0.1$, $t\in\{0,0.1,\dots,0.5\}$) with $\theta_J=-\pi/2$ and $\theta_h\in\{0,\pi/8,\pi/2\}$ with projector (solid) and standard unraveling (dashed).
  After each entangling gate we apply a two-qubit depolarizing channel of strength $\gamma\in\{10^{-3},10^{-2},10^{-1}\}$.
  Top: trajectory mean of $\langle Z_{106}\rangle$ over $N_{\mathrm{traj}}=100$ trajectories per unraveling method.
  Middle: empirical trajectory variance of $\langle Z_{106}\rangle$, implying sampling error $\mathrm{SE}=\sqrt{\mathrm{Var}/N_{\mathrm{traj}}}$.
  Bottom: average MPS bond dimension at each time step (soft SVD truncation $10^{-16}$, hard cap $\chi_{\max}=128$).
Black star: zero-noise extrapolated (noise-mitigated) experimental estimate of $\langle Z_{106}\rangle$ at $t=0.5$ from Ref.~\cite{kim2023evidence}.}
\label{fig:127q_plot}
\end{figure*}

\section{Discussion and outlook}
\label{sec:discussion}

We presented the circuit Tensor Jump Method (cTJM), a trajectory-based tensor network framework for simulating noisy quantum circuits that unifies local TDVP gate application with the TJM. For sparse Pauli-Lindblad hardware models, the jump hazard becomes state-independent and the dissipative contraction reduces to a global scalar that cancels upon renormalization. This removes auxiliary hazard estimators and makes layer-wise sampling stable and efficient while retaining pure state MPS scaling.

A central advantage of cTJM is its ability to represent hardware-connectivity noise, including correlated multi-qubit channels on non-adjacent subsets. In both the analog and projector unravelings, each collapse operator is (up to a scalar) of the form $a\,\id+b\,P$, which admits an exact bond-2 MPO independent of the separation of the support (\cref{sec:longrange-mpo}). Hence, long-range crosstalk terms can be applied without SWAP overhead and without inflating operator bond dimensions, enabling direct simulation of calibrated connectivity-induced errors beyond local noise models.

We analysed two variance-aware unravelings with complementary regimes. Analog unitary mixtures match the Pauli-Lindblad generator exactly under symmetric angle laws and reduce fluctuations in the weak-noise limit by replacing rare Pauli flips with frequent near-identity kicks. Projector jumps exploit unraveling freedom to induce absorbing-window behavior: when the readout anticommutes with all active channels, we obtain closed-form Bernoulli variance laws (\cref{theorem:projector-variance}) and observe strong suppression of both Monte Carlo variance and bond dimension growth. Numerically, this regime separation is visible in the 25-qubit XY quench and persists at hardware scale in the 127-qubit kicked-Ising benchmark, where variance remains bounded and entanglement growth is controlled at fixed truncation.

Looking forward, the explicit dependence of variance on local anticommute rates suggests hybrid, window-adaptive unravelings that switch between analog and projector sampling to minimize a combined cost proxy of variance and bond dimension or minimize trajectory entanglement by adaptive unraveling schemes. Further extensions include multi-jump-per-window sampling for larger hazards and broadening beyond Pauli-Lindblad noise via structured approximations while retaining connectivity-faithful correlated channels.

\section{Code Availability}

The algorithms and numerical experiments reported in this work were implemented in Julia and are available in the repository \cite{frohlich_yaqs_julia_2026}. In addition, an open-source Python implementation is available through the \emph{MQT-YAQS} package \cite{YAQS}, which is part of the \emph{Munich Quantum Toolkit} \cite{wille_mqt2024}.

\begin{acknowledgments}
This research was supported by the Einstein Research Unit (ERU) on quantum devices (WIAS) and joint work with the Technical University of Munich. Martin Eigel acknowledges partial financial support from the German Federal Ministry of Education and Research (BMBF) under grant agreement No.~13N17160 (Q-ROM - Quantum Read-Once-Memory: Verwandlung von klassischen Daten zu Quantenzuständen). The authors acknowledge the use of AI-based tools to support the preparation of this manuscript. All AI-assisted output was critically reviewed and verified by the authors. The authors assume responsibility for all content.
\end{acknowledgments}

\bibliographystyle{apsrev4-2}
\bibliography{bib}

\newpage
\onecolumngrid

\appendix
\section{Closed-form variance for analog two-point and Gaussian schemes (two-qubit IX/XI/XX test)}
\label{appendix:analog-variance-2pt-gauss}

We derive the across-trajectory variance of local $Z$-observables under the analog-unitary unraveling for the two-qubit “noise-only” benchmark used in the main text (Sec.~\ref{sec:variance-check}): equal physical rates $\gamma$ are applied on the Pauli strings $\{\,X\!\otimes\!I,\ I\!\otimes\!X,\ X\!\otimes\!X\,\}$, we initialize $\ket{00}$, insert identity layers between noise layers, and record $Z_i(t):=\langle Z_i\rangle_{\rm traj}(t)$ after each layer. Each anticommuting channel uses rate-free kicks $U_\theta=e^{i\theta P}$ with algorithmic rate density $\gamma_\theta=\lambda\,w(\theta)$ such that
$
\lambda\,s=\gamma,\quad s:=\mathbb E_w[\sin^2\theta]
$ and with  $w(\theta)=w(-\theta)$ being the symmetric angle law. \\

\paragraph*{Mean (scheme-independent)}
Let $P$ be a Pauli string with $\{P,Z_i\}=0$ and $U_\theta:=e^{i\theta P}$ the rate-free analog unitary.
Using the Euler's identity for matrices (which holds true because $P^2=\id$), the conjugation of $Z_i$ closes on a
two-dimensional subspace and yields
\begin{equation}
\label{eq:K_i_operator}
U_\theta^\dagger\,Z_i\,U_\theta = \cos(2\theta)\,Z_i-\sin(2\theta)\,K_i,
\quad
K_i := \tfrac{i}{2}[P,Z_i],
\end{equation}
i.e.\ a rotation by angle $2\theta$ around the $P$-axis (the factor $2$ appears because we use $e^{i\theta P}$, not $e^{i(\theta/2)P}$).

Since $\sin(2\theta)$ is odd, $\mathbb E_w[\sin(2\theta)]=0$, while
\[
\mathbb E_w[\cos(2\theta)] \;=\; \mathbb E_w[1-2\sin^2\theta] \;=\; 1-2s.
\]
Hence a single averaged kick maps the mean of $Z_i$ to $(1-2s)\,\mathbb E[Z_i]$.

For an anticommuting channel $p$ with state-independent analog hazard $\lambda_p$,
the conditional update of the mean is
\[
\mathbb{E}[Z_i(t{+}\delta t)\mid Z_i(t)] =\Big(1+\lambda_p\delta t[\mathbb{E}_w(\cos2\theta)-1]\Big)Z_i(t)+o(\delta t).
\]
Taking expectations and subtracting $\mathbb E[Z_i(t)]$, then dividing by $\delta t$ and sending $\delta t\to0$ yields
\[
\frac{d}{dt}\,\mathbb E[Z_i(t)] \;=\; \lambda_p\!\left(\mathbb E_w[\cos(2\theta)]-1\right)\mathbb E[Z_i(t)]
\;=\; -2\lambda_p s\,\mathbb E[Z_i(t)],
\]
so $\mathbb E[Z_i(t)]=e^{-2\lambda_p s\,t}\,\mathbb E[Z_i(0)]=e^{-2\lambda_p s\,t}$ for that channel. Note that in our test $\mathbb E[Z_i(0)]=1$.\\
For $Z_i$ in the IX/XI/XX test, there are exactly two anticommuting channels (the local $X_i$ and the correlated $X_0X_1$), each with physical rate $\gamma$.
Analog generator matching enforces $\lambda_p s=\gamma$ for each $p$, hence
\begin{equation}
\mathbb E[Z_i(t)]
=\exp\!\Big(t\sum_{\text{anti }p}\lambda_p(\mathbb E_w[\cos 2\theta]-1)\Big)\mathbb E[Z_i(0)]
=\exp\!\Big(-2t\sum_{\text{anti }p}\lambda_p s\Big)
=\exp(-4\gamma t),
\end{equation}
where we have used the identities $\mathbb{E}_w[\cos(2\theta)]-1=-2s$ with $s:=\mathbb{E}_w[\sin^2\theta]$, channel-wise generator matching $\lambda_p s=\gamma_p$, hence $\sum_{\text{anti }p}\lambda_p s=\sum_{\text{anti }p}\gamma_p$ and $\sum_{\text{anti }p}\gamma_p=\gamma+\gamma=2\gamma$. Therefore $\big(\mathbb E[Z_i(t)]\big)^2=e^{-8\gamma t}$, independent of the angle law $w$.

\paragraph*{Second moment and variance}
In the following we derive the second moment $\mathbb E\!\big[Z_i(t)^2\big]$ and the variance $\operatorname{Var}[Z_i(t)]$.
The single–site estimator of a single trajectory is $Z_i(t)=\langle\psi_t|Z_i|\psi_t\rangle.$
Its square can be written as a two–copy expectation:
\[
Z_i(t)^2 = \mathrm{Tr}\left[(Z_i\otimes Z_i)\big(|\psi_t\rangle\langle\psi_t|\big)^{\otimes 2}\right].
\]
Averaging over trajectories defines
\[
\sigma(t):=\mathbb E\big[|\psi_t\rangle\langle\psi_t|^{\otimes 2}\big],\quad
\mathbb E\big[Z_i(t)^2\big]=\mathrm{Tr}\left[(Z_i\otimes Z_i)\sigma(t)\right].
\]
Under a quantum-trajectory unraveling, $\sigma(t)$ obeys a Markovian master equation generated by a two-copy Liouvillian, obtained by replacing each single-copy jump ($L$) by $L\otimes L$ (and likewise for the effective no-jump drift) \cite{salgado2002lindbladianevolutionselfadjointlindblad}. For our analog unitary collapse operators $L_{p,\theta}=\sqrt{\lambda_p w(\theta)}U_{p,\theta}$ with $U_{p,\theta}=e^{i\theta P_p}$, the two-copy generator on $\sigma$ is
\begin{equation}
\mathcal L^{(2)}(\sigma)
=\sum_{p}\lambda_p\int w(\theta)\Big[(U_{p,\theta}\otimes U_{p,\theta})\sigma(U_{p,\theta}^\dagger\otimes U_{p,\theta}^\dagger)-\sigma\Big]d\theta,
\end{equation}
where we have used the entry-wise tensor product of maps, specialised to conjugation $\mathrm{Ad}_U(X)=UXU^\dagger$ gives $(\mathrm{Ad}_U\otimes \mathrm{Ad}_U)(\sigma)=(U \otimes U)\sigma(U^\dagger\otimes U^\dagger)$ \cite{watrous2018theory}.

Now we fix qubit $i$ and a single anticommuting Pauli string $P$ (so $\{Z_i,P\}=0$) and recall $K_i$ from \cref{eq:K_i_operator}, for which

\begin{equation}
U_{p,\theta}^\dagger\,K_i\,U_{p,\theta} = \sin(2\theta)\,Z_i \;+\; \cos(2\theta)\,K_i
\end{equation}
holds true.

\medskip
On two copies, this single channel preserves the 3-dimensional subspace
\[
\mathcal V_p \;:=\; \mathrm{span}\big\{\,Z_i\!\otimes\! Z_i,\; S_p,\; K_i\!\otimes\! K_i\,\big\},
\quad
S_p \;:=\; Z_i\!\otimes\!K_i \;+\; K_i\!\otimes\!Z_i,
\]
because $U_{p,\theta}\!\otimes\!U_{p,\theta}$ applies the same $2\theta$-rotation to each factor.
Writing $c:=\cos(2\theta)$, $s:=\sin(2\theta)$, the single-kick action $U_{p,\theta}\!\otimes\!U_{p,\theta}$ (before averaging over $\theta$) is
\begin{equation}
\label{eq:ZZ-S-KK-singlekick-plus}
\begin{aligned}
  Z\!\otimes\!Z &\;\mapsto\; c^2\,(Z\!\otimes\!Z)\;-\;c s\,S_p\;+\;s^2\,(K\!\otimes\!K),\\[2pt]
  S_p &\;\mapsto\; -2c s\,(Z\!\otimes\!Z)\;+\;(c^2-s^2)\,S_p\;+\;2c s\,(K\!\otimes\!K),\\[2pt]
  K\!\otimes\!K &\;\mapsto\; s^2\,(Z\!\otimes\!Z)\;+\;c s\,S_p\;+\;c^2\,(K\!\otimes\!K).
\end{aligned}
\end{equation}

\medskip
Averaging over an even angle law $w(\theta)=w(-\theta)$ kills all $c s$-terms ($\mathbb{E}_w[c s]=0$).
Using the standard power reducing formulas
\[
\mathbb{E}_w[c^2]=\tfrac12\big(1+\mathbb{E}_w[\cos(4\theta)]\big),\quad
\mathbb{E}_w[s^2]=\tfrac12\big(1-\mathbb{E}_w[\cos(4\theta)]\big),
\]
the $\theta$-averaged one-kick map on $\mathcal V_p$ in the ordered basis $(Z\!\otimes\!Z,\; S_p,\; K\!\otimes\!K)$ is

\begin{equation}
\label{eq:Mp-plus}
\mathsf M_p \;=\;
\begin{pmatrix}
  \mathbb{E}[c^2] & 0 & \mathbb{E}[s^2]\\[2pt]
  0 & \mathbb{E}[c^2{-}s^2] & 0\\[2pt]
  \mathbb{E}[s^2] & 0 & \mathbb{E}[c^2]
\end{pmatrix}
\;=\;
\begin{pmatrix}
  \frac{1+m_4}{2} & 0 & \frac{1-m_4}{2}\\[4pt]
  0 & m_4 & 0\\[4pt]
  \frac{1-m_4}{2} & 0 & \frac{1+m_4}{2}
\end{pmatrix},
\quad
m_4:=\mathbb{E}_w[\cos(4\theta)].
\end{equation}
Here eigenvectors are eigenoperators of $\mathsf M_p$: since the $2\times2$ block on \\ $\mathrm{span}\{{Z\otimes Z,K\otimes K}\}$ has the form $\bigl[
\begin{smallmatrix}a&b\\ b&a
\end{smallmatrix}\bigr]$ with $(a=\tfrac{1+m_4}{2}), (b=\tfrac{1-m_4}{2})$, the symmetric/antisymmetric combinations $v_0:=Z\otimes Z+K\otimes K$ and $v_-:=Z \otimes Z-K \otimes K$ diagonalize it, with eigenvalues $1$ and $m_4$, respectively; the cross-mode $S_p$ decouples and also has eigenvalue $m_4$. Passing from the one-kick map to continuous time with Poisson hazard $\lambda_p$ yields the generator $\lambda_p(\mathsf M_p-\mathsf I)$; an eigenvalue $\mu$ thus contributes a factor $e^{\lambda_p(\mu-1)t}$, i.e., decay rate $\lambda_p(1-\mu)$. Hence $v_-$ and $S_p$ decay at rate $\lambda_p(1-m_4)$, while $v_0$ is conserved.

It is convenient write
\begin{equation*}
1-m_4 \;=\; 1-\mathbb{E}_w[\cos(4\theta)]
\end{equation*}
as
\begin{equation}
\label{eq:b_p_definition}
b_p \;:=\; \frac{\lambda_p}{2}\Big(1-\mathbb{E}_w[\cos(4\theta)]\Big)
\;=\; \frac{\gamma_p}{s}\cdot\frac{1-\mathbb{E}_w[\cos(4\theta)]}{2},
\quad (\lambda_p s=\gamma_p).
\end{equation}
With the initial state $\sigma(0)=\ket{00}\!\bra{00}^{\otimes 2}$, along $\mathcal V_p$ we have
\begin{equation*}
Z\!\otimes\! Z \;=\; \tfrac12 v_0+\tfrac12 v_-
\;\;\Longrightarrow\;\;
e^{\,t\lambda_p(\mathsf M_p-\mathsf I)}(Z\!\otimes\! Z)
\;=\; \tfrac12 v_0+\tfrac12 e^{-2b_p t}\,v_-.
\end{equation*}
Taken the expectation value of $Z\!\otimes\! Z$ yields the single–channel second moment
\begin{equation*}
\mathbb{E}_p\!\big[Z_i(t)^2\big]
\;=\; \tfrac12 \;+\; \tfrac12\,e^{-2b_p t}.
\end{equation*}

In this example $Z_i$ has exactly two anticommuting channels with identical parameters, such that $b_1 = b_2 =: b$. Consequently we get:
\begin{equation}
\label{eq:second-moment-master}
  \mathbb E\!\big[Z_i(t)^2\big]
  \;=\; \tfrac14 \;+\; \tfrac12\,e^{-2bt} \;+\; \tfrac14\,e^{-4bt},
  \quad
  \operatorname{Var}[Z_i(t)]
  \;=\; \tfrac14 \;+\; \tfrac12\,e^{-2bt} \;+\; \tfrac14\,e^{-4bt} \;-\; e^{-8\gamma t}.
\end{equation}

\paragraph*{Two-point law}
For $\theta\in\{\pm\theta_0\}$ with equal weight,
$
s=\sin^2\theta_0,\
\mathbb E[\cos(4\theta)]=\cos(4\theta_0)=1-8s(1-s).
$
Plugging into \eqref{eq:b_p_definition} gives
\begin{equation}
b_{\rm 2pt}
=\frac{\gamma}{s}\cdot\frac{1-(1-8s(1-s))}{2}
=4\gamma(1-s).
\end{equation}
Thus
\begin{equation}
\label{eq:var-2pt}
\operatorname{Var}_{\rm 2pt}[Z_i(t)]
=\tfrac14+\tfrac12\,e^{-8\gamma(1-s)t}+\tfrac14\,e^{-16\gamma(1-s)t}-e^{-8\gamma t}.
\end{equation}

\paragraph*{Gaussian law}
For $\theta\sim\mathcal N(0,\sigma^2)$, $s=\tfrac12(1-e^{-2\sigma^2}),\ \mathbb E[\cos(4\theta)]=e^{-8\sigma^2}=(1-2s)^4$.
Hence
\begin{equation}
b_{\rm gauss}
=\frac{\gamma}{s}\cdot\frac{1-(1-2s)^4}{2}
=\gamma\,\frac{1-e^{-8\sigma^2}}{1-e^{-2\sigma^2}}
=\gamma\,\big(1+e^{-2\sigma^2}+e^{-4\sigma^2}+e^{-6\sigma^2}\big),
\end{equation}
and
\begin{equation}
\label{eq:var-gauss}
\operatorname{Var}_{\rm gauss}[Z_i(t)]
=\tfrac14+\tfrac12\,e^{-2b_{\rm gauss} t}+\tfrac14\,e^{-4b_{\rm gauss} t}-e^{-8\gamma t}.
\end{equation}

\end{document}